\documentclass[a4paper,11pt]{article}

\usepackage{amsmath,amssymb,amsthm,bbm,graphicx,url,color,bm}
\usepackage[round]{natbib}

\theoremstyle{plain}  
\newtheorem{thm}{Theorem}[section] 
\newtheorem{lem}[thm]{Lemma} 
\newtheorem{prop}[thm]{Proposition}

\theoremstyle{definition} 
\newtheorem{defi}[thm]{Definition}
 
\newtheorem{ex}[thm]{Example}

\theoremstyle{remark}

\newtheorem*{rem}{Remark} 

\newcommand{\R}{\mathbb{R}} 
\newcommand{\N}{\mathbb{N}}
\newcommand{\E}{\mathbb{E}} 
\newcommand{\Q}{\mathbb{Q}}

\newcommand{\A}{\mathcal{A}} 
\newcommand{\cB}{\mathcal{B}} 

\newcommand{\cF}{\mathcal{F}} 
 
\newcommand{\cL}{\mathcal{L}} 
\newcommand{\cN}{\mathcal{N}} 
\newcommand{\cT}{\mathcal{T}}

\newcommand{\one}{\mathbbm{1}}

\newcommand{\Var}{\mathrm{var}}
\newcommand{\supp}{\operatorname{supp}}
\newcommand{\logit}{\operatorname{logit}}
\newcommand{\CRPS}{\operatorname{CRPS}}
\newcommand{\DSS}{\operatorname{DSS}}

\begin{document}

\title{Cross-calibration of probabilistic forecasts}
\author{Christof Str\"ahl\thanks{E-mail: christof.straehl@stat.unibe.ch} \and Johanna F.~Ziegel\thanks{E-mail: johanna.ziegel@stat.unibe.ch} }
\date{University of Bern, Switzerland}
\maketitle

\begin{abstract}
When providing probabilistic forecasts for uncertain future events, it is common to strive for calibrated forecasts, that is, the predictive distribution should be compatible with the observed outcomes. Several notions of calibration are available in the case of a single forecaster alongside with diagnostic tools and statistical tests to assess calibration in practice. Often, there is more than one forecaster providing predictions, and these forecasters may use information of the others and therefore influence one another. We extend common notions of calibration, where each forecaster is analysed individually, to notions of cross-calibration where each forecaster is analysed with respect to the other forecasters in a natural way. It is shown theoretically and in simulation studies that cross-calibration is a stronger requirement on a forecaster than calibration. Analogously to calibration for individual forecasters, we provide diagnostic tools and statistical tests to assess forecasters in terms of cross-calibration. The methods are illustrated in simulation examples and applied to probabilistic forecasts for inflation rates by the Bank of England.\end{abstract}

\section{Introduction}
\label{sec:crosscalibration}

In the past decades probabilistic forecast, specifying a complete predictive probability distribution for an uncertain future event, have replaced point forecasts in a number of applications including weather forecasting, climate predictions and economics; see \citet{GneitingKatzfuss2014} for a recent overview. \citet{MurphyWinkler1987,GneitingBalabdaouETAL2007} have formulated the guiding principle for a probabilistic forecast to ``Maximize sharpness subject to calibration''. Calibration refers to the statistical compatibility of the forecasts and the observations. Sharpness, on the other hand, is a property that concerns the forecast only. Roughly speaking, a forecast is sharper the more concentrated the distribution is, with point forecasts as a limiting case. \citet{GneitingBalabdaouETAL2007} have formulated their principle in order to pick the ``better'' of two calibrated forecasts. While it is generally acknowledged that forecasts should be calibrated \citep{Dawid1984,DieboldGuntherETAL1998}, it is not universally accepted that it is necessary to consider sharpness as a further criterion for forecast evaluation \citep{MitchellWallis2011}. 

In this manuscript we are concerned with calibration only. However, we consider a framework where several forecasters issue competing forecasts. We propose concepts of cross-calibration in order to formalize the influence of forecasters amongst each other and with respect to the observations. Essentially, a cross-calibrated forecaster not only uses her own information optimally but also incorporates the information of the competing forecasters in an optimal way. The notions we propose refine the existing notions of calibration of \citet{GneitingRanjan2013}. Furthermore, we extend their prediction space setting to allow for serial dependence which is the usual situation in forecasting applications. We are able to extend the result of \citet{DieboldGuntherETAL1998} of uniformity and independence of probability integral transform (PIT) values to our general framework.

Notions of cross-calibration have previously been considered  in the literature for binary or categorical outcomes. \citet{Al-NajjarWeinstein2008} consider a test which an uninformed forecaster cannot pass with high probability when an informed forecaster is present. The notion of cross-calibration by \citet{FeinbergStewart2008} takes into account that several forecasters may influence each other, and the one with the largest information set should be preferred. In this paper we generalize the cross-calibration notions of \citet{FeinbergStewart2008} to forecasts of real valued outcomes including diagnostic tools and statistical tests to assess cross-calibration in applications. 

The cross-calibration test suggested by \citet{FeinbergStewart2008} uses the
following framework, which we review here only in the case of two forecasters for simplicity. Let $\Omega = \{(\omega_t)_{t=0,1,\ldots}|\,\omega_t \in \{\,0,1\}\}$ denote the space of all possible
realizations and let $n > 4$ be an integer. Divide the interval $[ 0,1]$ into $n$ equal closed subintervals
$[0,1/n],\dots,[(n-1)/n,1]$. At time $t$, forecaster $j$, $j = 1,2$, makes a prediction which is given as an interval $I_t^j \in \{[0,1/n],\dots,[(n-1)/n,1]\}$. It gives bounds on the predictive probability that the next
realization $\omega_t$ is equal to one. The cross-calibration
test is defined over the sequence of forecast-observation triples
$(I_t^1,I_t^2,\omega_t)_{t=0}^\infty$. For any pair $\ell = (\ell_1,\ell_2) \in \{1,\ldots,n\}^2$ and any time $T$, let
\begin{displaymath}
\nu_T^\ell = \sum_{t=0}^T \one\Big(I_t^1 = \Big[\frac{\ell_1-1}{n},\frac{\ell_1}{n}\Big],I_t^2 =\Big[\frac{\ell_2-1}{n},\frac{\ell_2}{n}\Big]\Big),
\end{displaymath}
which is the number of times up to $T$, that the forecasting profile $\ell$ is chosen. For
$\nu_T^\ell > 0$, the frequency of realizations equal to one conditional on the forecasting
profile is given by
\begin{displaymath}
  f_T^\ell = \frac{1}{\nu_T^l} \sum_{t=0}^T \omega_t \one\Big(I_t^1 = \Big[\frac{\ell_1-1}{n},\frac{\ell_1}{n}\Big],I_t^2 =\Big[\frac{\ell_2-1}{n},\frac{\ell_2}{n}\Big]\Big).
\end{displaymath}
A forecaster $j$ passes the cross-calibration
test at the outcome $(I_t^1,I_t^2,\omega_t)_{t=0}^\infty$ if
\begin{displaymath}
  \limsup_{T \rightarrow \infty} \left| f_T^\ell - \frac{2\ell_j-1}{2n}\right| \le \frac{1}{2n}
\end{displaymath}
for every $\ell$ satisfying $\lim_{T \rightarrow \infty} \nu_T^\ell = \infty$.

It is shown in \citet{FeinbergStewart2008} that a forecaster who is aware of the distribution of
$(\omega_t)_{t=0}^\infty$ passes the cross-calibration test with probability one, no matter
which strategy the other forecaster uses.
From a theoretical point of view, this is an interesting result. However, testing empirically
if a forecast is cross-calibrated is rather difficult. The problem is that already if
$n=5$, there are 25 forecasting profiles to consider. For each of these profiles the
empirical frequency conditioned on that profile should lie inside the predicted interval
of the cross-calibrated forecaster. But some profiles are hardly ever predicted and
therefore the number of observations needs to be very large. We illustrate this problem with the following simple simulation example, which has been implemented in R \citep{R-Develop2008} like all further simulations in this paper. 

\begin{ex}
\label{ex:crosscalibration}
\begin{table}[t]
  \centering
  \begin{tabular}{|c|c|} \hline
T & Monte-Carlo power \\ \hline
$10^4$ & 0.112 \\
$5 \cdot 10^4$ & 0.254 \\
$10^5$ & 0.333 \\
$5 \cdot 10^5$ & 0.699 \\
$10^6$ & 0.847 \\ 
$5 \cdot 10^6$ & 0.994 \\ \hline
      \end{tabular}
  \caption{Monte-Carlo power of detecting cross-calibration for different time periods $T$.  \label{tab:crosscalibration}}
\end{table}
In this example we consider the setting of the cross-calibration test described above. Let $B_t$, $C_t$, $t = 0,\dots,T$ be independent beta random variables with parameters $(3,5)$ for $B_t$ and $(2,1.5)$ for $C_t$. We simulate a (finite) stochastic process $(\omega_t)_{t=0}^T$, where $\omega_t$ is conditionally Bernoulli distributed with probability $(B_t + C_t)/2$. Let $n=5$. The first forecaster predicts at each time $t$ the interval $I_t^1$ which contains the value $(B_t + C_t)/2$, which is the
probability that the realization $\omega_t$ is one. The second forecaster predicts the interval $I_t^2$ which contains the value $C_t$. Therefore, we expect that the first forecaster
is cross-calibrated with respect to the second forecaster and should pass the test. The first forecaster passes the test if for all forecasting profiles $\ell = (\ell_1,\ell_2)$ where $\nu_T^\ell$ is positive $f_T^\ell$ lies in $[(\ell_1-1)/n,\ell_1/n]$. In Table
\ref{tab:crosscalibration} we see the result. For several $T$  we performed
the test 1000 times. The second column of the table gives the Monte-Carlo power of the
test, that is, how often the test detected the cross-calibration of the first forecaster divided by $T$.
\end{ex}

Table \ref{tab:crosscalibration} shows that already for this rather simple
example, the sample size $T$ needs to be large in order to come close to the theoretically predicted power of one. Furthermore, the test is only applicable to probabilistic forecasts for binary outcomes. 
The goal of this paper is to extend the notion of cross-calibration to probabilistic forecasts of real valued quantities, and present methodology to empirically assess cross-calibration for serially dependent forecast-observation tuples. We have chosen to work in the framework of prediction spaces as introduced by \citet{GneitingRanjan2013}, and extend it to allow for serial dependence. 

The paper is organized as follows. In Section \ref{sec:framework} we review and extend the notion of a prediction space and generalize the notions of calibration for individual forecasters to multiple forecasters. We introduce diagnostic tools for checking cross-calibration and illustrate their usefulness in a simulation study in Section \ref{sec:diagnostic}. In Section \ref{sec:binary.outcomes} we treat the special case of binary outcomes and relate our work to the existing results of \citet{FeinbergStewart2008}. Statistical tests for cross-calibration are derived in Section \ref{sec:tests}. We analyse the Bank for England density forecasts for inflation rates in Section \ref{sec:data}. Finally, the paper concludes with a discussion in Section \ref{sec:discussion}.

\section{Notions of cross-calibration}
\label{sec:framework}

\citet{GneitingRanjan2013} introduced the notion of a \emph{prediction space} as follows.

\begin{defi}[one-period prediction space]
\label{def:prediction.space}
  Let $k \ge 1$ be an integer. A \emph{prediction space} is a probability space $(\Omega, \mathcal{A}, \Q)$
  together with sub-$\sigma$-algebras $\mathcal{A}_1,\ldots,\mathcal{A}_k \subset
  \mathcal{A}$. The elements of $\Omega$ are tuples of the form $(F_1,\ldots,F_k,Y,V)$
  such that, for $i = 1,\ldots,k$, $F_i$ is a CDF-valued random quantity that is measurable with respect to $\mathcal{A}_i$ \footnote{That is, for all finite collections $x_1,\ldots,x_n \in \mathbb{R}$, $B_1,\ldots,B_n \in \mathcal{B}(\mathbb{R})$, the event $\{F_i(x_j) \in B_j$ for $j = 1,\ldots,n \} \in \mathcal{A}_i$.}, $Y$ is a real-valued random variable, and  $V$ is a uniformly distributed random variable on $[0,1]$, independent of $\mathcal{A}_1,\ldots,\mathcal{A}_k,Y$.
\end{defi}

The integer $k$ corresponds to the number of forecasters. The $\sigma$-algebra $\mathcal{A}_i$ can be seen as the information set available to forecaster $i$. The random variable $Y$ is the observation,  the random variable $V$ is needed for technical reasons. It allows to define the probability integral transform (PIT) in Definition \ref{def:pit} below. 

We term the prediction space proposed by \citet{GneitingRanjan2013} a \emph{one-period} prediction space as it is only concerned with predictions for an outcome $Y$ at one time point. While this framework is sufficient to define various notions of calibration and cross-calibration of forecasters in principle, a statistical analysis of calibration is only possible if we can assume that we have a sequence $(F_{1,n},\dots,F_{k,n},Y_n,V_n)_{1 \le n \le N}$ of independent forecast-observation tuples. This assumption is unrealistic in most forecasting situations. Therefore, we propose to extend the prediction space setting, allowing for serial dependence as follows. 

\begin{defi}[prediction space for serial dependence]
  \label{def:prediction.space.SD}
Let $k \geq 1$ be an integer. A \emph{prediction space for serial dependence} is a probability space $(\Omega, \A, \Q )$ together with filtrations $(\A_{1,t})_{t \in \N}, \ldots (\A_{k,t})_{t \in \N} \subset \A$. The elements of $\Omega$ are sequences of tuples of the form $(F_{1,t}, \ldots, F_{k,t}, Y_{t+1}, V_t)_{t \in \N}$, where
$(Y_t)_{t \in \N}$ is a sequence of real-valued random variables, and $(V_t)_{t \in \N}$ is an iid sequence of standard uniform random variables that is independent of everything else. Let $\mathcal{T}_t = \sigma(Y_s \mid s \le t)$ be the $\sigma$-algebra generated by the observations until time $t$. For all $t \in \N$ and $i=1,\dots,k$, $F_{i,t}$ is a CDF-valued random quantity that is $\sigma(\mathcal{A}_{i,t},\mathcal{T}_t)$-measurable. We assume that, for all $t \in \N$, $m \ge 1$,
\begin{equation}\label{eq:serialass}
\mathcal{L}(Y_{t+1} \mid \mathcal{A}_{1,t+m},\dots,\mathcal{A}_{k,t+m},\mathcal{T}_t) = \mathcal{L}(Y_{t+1} \mid \mathcal{A}_{1,t},\dots,\mathcal{A}_{k,t},\mathcal{T}_t),
\end{equation}
where $\mathcal{L}(X \mid \mathcal{G})$ denotes the conditional law of a random variable $X$ with respect to the $\sigma$-algebra $\mathcal{G}$.
\end{defi}

The notation in Definition \ref{def:prediction.space.SD} is chosen such that $\A_{i,t}$ encodes the information of the $i$-th forecaster $F_{i,t}$ at time $t$ to predict the outcome $Y_{t+1}$ at the next time point. Additionally, all forecasters $F_{1,t},\dots,F_{k,t}$ have access to the past realizations of $Y_t$ in principle, that is, to the information contained in $\cT_t$. This means, we have separated the information of forecaster $F_i$ into two parts, the information of past realizations of the outcome $\cT_t$, that is available to all forecasters, and a personal information set $\A_{i,t}$ that she acquires (partially) from other sources. Condition \eqref{eq:serialass} formalizes that information from other sources about the outcome at time point $t+1+m$ should not influence the outcome $Y_{t+1}$ at time points $t+1$. A sufficient condition for \eqref{eq:serialass} to hold is that $\A_{i,t+m} = \sigma(\A_{i,t},\cB_{i,t+m})$ with $\cB_{i,t+m}$ independent of $\A_{i,t}$ and $\cT_t$.

Let us illustrate this point in the context of weather forecasting. Suppose a numerical weather prediction system is used to calculate the state of the atmosphere to help us predict temperature tomorrow. Condition \eqref{eq:serialass} means that if we let the numerical system run longer to give us also information about the atmosphere the day after tomorrow, this will have no influence on what temperature is realized tomorrow.

All further statements are within the prediction space setting and expressions such as \emph{almost surely} are with respect to the probability measure $\Q$. In the prediction space for serial dependence, $F_{i,t}$ is termed \emph{ideal} with respect to $\A_{i,t}$ if 
\begin{displaymath}
  F_{i,t} = \cL(Y_{t+1} \vert \A_{i,t},\cT_t) \quad \text{almost surely}.
\end{displaymath}
In the case of independent forecast-observation tuples, we recover the definition of an ideal forecaster of \citet{GneitingRanjan2013}, that is, in the one-period prediction space setting, $F_i$ is \emph{ideal} with respect to $\mathcal{A}_i$ if 
\[
F_i = \mathcal{L}(Y|\mathcal{A}_i) \quad \text{almost surely};
\]
see also \citet{Tsyplakov2011,Tsyplakov2013}. We generalize this notion as follows.
\begin{defi}[cross-ideal]
\label{def:cross-ideal}
In the prediction space setting for serial dependence, we call $F_{i,t}$ \emph{cross-ideal} with respect to $\A_{1,t}, \ldots ,\A_{k,t}$ if 
\begin{equation}\label{eq:crossideal}
  F_{i,t} = \mathcal{L}(Y_{t+1} \vert \mathcal{A}_{1,t}, \ldots,\mathcal{A}_{k,t}, \cT_t) \quad \text{almost surely.}
\end{equation}
\end{defi}
A cross-ideal forecaster does not only use her own information optimally but also the
information available to the other forecasters. In fact, at time $t$, her information $\mathcal{A}_{i,t}$ contains all relevant information of all the forecasters because $F_{i,t}$ is $\sigma(\mathcal{A}_{i,t},\cT_t)$-measurable and hence by \eqref{eq:crossideal}, also $\mathcal{L}(Y_{t+1}|\mathcal{A}_{1,t}, \ldots,\mathcal{A}_{k,t}, \cT_t)$ is $\sigma(\mathcal{A}_{i,t},\cT_t)$-measurable, implying that $\mathcal{L}(Y_{t+1}|\mathcal{A}_{1,t}, \ldots,\mathcal{A}_{k,t}, \cT_t) = \mathcal{L}(Y_{t+1}|\mathcal{A}_{i,t}, \cT_t)$. Therefore, each cross-ideal forecaster is ideal, whereas the converse does not hold in general; see Examples \ref{ex:cross-ideal} and \ref{ex:marginal.calibration}. The above argument shows more generally the following proposition. 
\begin{prop}\label{prop:sub-cross-ideal}
For some $t \in \N$, let $F_{1,t},\dots,F_{k,t}$ be forecasters with information sets $\mathcal{A}_{1,t},\dots,\mathcal{A}_{k,t}$ in a prediction space for serial dependence. If $F_{1,t}$ is cross-ideal with respect to $\mathcal{A}_{1,t},\dots,\mathcal{A}_{k,t}$, then it is also cross-ideal with respect to $\mathcal{A}_{1,t},\mathcal{A}_{i_2,t},\dots,\mathcal{A}_{i_m,t}$,  where $\{i_2,\dots,i_m\}\subset \{2,\dots,k\}$.
\end{prop}

For clarity, we have chosen to illustrate the notions of cross-ideal forecasters (or cross-calibrated forecasters; see Definition \ref{def:cross.calibration}) with independent forecast-observation tuples, or, in other words, in the one-period prediction space setting of \citet{GneitingRanjan2013} dropping the time index $t$. This is natural, as the notions of calibration are essentially one-period concepts, and make no use of assumption \eqref{eq:serialass}. The purpose of assumption \eqref{eq:serialass} will become clear in Theorem \ref{thm:serial.dependence} below where we generalize the result of \citet{DieboldGuntherETAL1998} on uniformity and independence of PIT values.

\begin{ex}\label{ex:cross-ideal} Let $\nu$ be uniformly distributed on $(5,20)$ and, conditionally on $\nu$, let $\sigma^2$ have an inverse chi-squared distribution with $\nu$ degrees of freedom. Conditional on $\nu$ and $\sigma$, the outcome $Y$ is normally distributed with mean zero and variance $\sigma^2$, and we consider two forecasters, a normally distributed forecaster $F_{1} = \cN(0,\sigma^2)$ and a t-distributed forecaster $F_{2} = t_{\nu}$. This example is constructed such that $F_1$ has the full information about the distribution of the outcome $Y$, whereas $F_2$ only knows the prior distribution of $\sigma^2$. We have that $F_1$ and $F_2$ are both ideal with respect to to their information sets $\mathcal{A}_1 = \sigma(\sigma^2)$ and $\mathcal{A}_2 = \sigma(\nu)$, respectively, but only $F_1$ is cross-ideal with respect to $\mathcal{A}_1,\mathcal{A}_2$. 

More specifically, the predictive density function $f_1(\cdot|\sigma^2)$ of $F_1$ is a normal density with variance $\sigma^2$, and the predictive density function $f_2$ of $F_2$ is
\begin{equation}\label{eq:f2}
  f_2(x|\nu) = \int_0^{\infty} f_1(x|s) g(s|\nu) \, ds = \frac{\Gamma\left(\frac{\nu+1}{2}\right)}{\sqrt{\nu\pi}\Gamma(\nu/2)}
  \left(1 + \frac{x^2}{\nu}\right)^{-\frac{\nu +1}{2}},
\end{equation}
where $g(\cdot|\nu)=(\nu/2)^{\nu/2}s^{\nu/2-1}\exp\{-\nu/(2s)\}/\Gamma(\nu/2)$ is the density function of an inverse chi-squared distribution with $\nu$ degrees of freedom. The right hand side of \eqref{eq:f2} is the density of a t-distribution. Equation \eqref{eq:f2} holds because for a normal likelihood with known mean, the inverse chi-squared distribution is a conjugate prior of a t-distributed posterior distribution. Therefore, we see that $F_1$ is cross-ideal with respect to $\mathcal{A}_1,\mathcal{A}_2$. It is clear that $F_2$ is not cross-ideal with respect to $\mathcal{A}_1,\mathcal{A}_2$. We will come back to this example throughout the paper. 
\end{ex}

The most prominent diagnostic tool for checking calibration empirically is the probability integral transform (PIT) \citep{Dawid1984,DieboldGuntherETAL1998}. 

\begin{defi}[PIT]
\label{def:pit} Let $F$ be a (possibly random) CDF, $X$ be a random variable and $V$ a standard uniform random variable independent of $F$ and $X$. We define
  \[Z^X_{F} = F(X -) + V  \{F(X)-F(X -)\}, \]
where $F(y -) = \lim_{x \uparrow y}F(x)$.
In the prediction space setting, the random variable $Z_{i,t} := Z_{F_{i,t}}^{Y_{t+1}}$ is called the \emph{probability integral transform (PIT)} of the $i$-th forecaster $F_{i,t}$.
\end{defi}
The PIT $Z_{i,t}$ is a random variable with values in $[0,1]$. If $F$ is deterministic and $X \sim F$, then $Z_F^X$ is uniformly distributed and $F^{-1}(Z_F^X) = X$ almost surely, where $F^{-1}$ is the quantile function of $F$; see for example \citet{Ruschendo2009}. Based on the PIT we introduce the following notions of cross-calibration.

\begin{defi}[cross-calibration]
\label{def:cross.calibration}
Let $F_{1,t},\dots,F_{k,t}$ be forecasters in a prediction space for serial dependence. Let $\{i_1,\dots,i_m\} \subset \{1,\dots,k\}$.
  \begin{enumerate}
 \item  \label{def:prob.cross.cal}
The forecast $F_{1,t}$ is \emph{cross-calibrated} with respect to $F_{i_1,t}$, \dots, $F_{i_m,t}$ if
\begin{displaymath}
  \mathcal{L}(Z_{1,t}|\, F_{i_1,t}, \dots, F_{i_m,t}, \cT_t) = \mathcal{U}([0,1]), \quad \text{almost surely,}
\end{displaymath}
where $\mathcal{U}([0,1])$ denotes a uniform distribution on $[0,1]$.\footnote{To be precise, the left hand side is a Markov kernel $\kappa:\Omega \times \mathcal{B}(\mathbb{R}) \to [0,1]$, which is required to be constant in $\omega \in \Omega$ and equal to the Lebesgue measure on $[0,1]$.}
 \item \label{def:marginal.cross-calibration}
For $1 \le j \le k$, $F_{1,t}$ is \emph{marginally cross-calibrated} with respect to $F_{j,t}$ if
\begin{displaymath}
\E_{\Q}{F_{j,t}(y)} = \E_{\Q}{\one\{F_{j,t}^{-1}(Z_{1,t}) \le y\}},
\end{displaymath}
for all $y \in \mathbb{R}$.
\end{enumerate}
\end{defi}

For brevity, we sometimes speak of cross-calibration with respect to $\{i_1,\dots,i_m\}$ instead of $F_{i_1,t}$, \dots, $F_{i_m,t}$.
Our definitions are natural generalizations of the notions of calibration for individual forecasters in \citet[Definition 2.6]{GneitingRanjan2013}, which we recall here fore ease of comparison.

\begin{defi}[calibration]
\label{def:calibration.dispersion}
Let $F$ be a forecaster in a one-period prediction space.
  \begin{enumerate}
    \item \label{def:prob.calibration}
    The forecast $F$ is \emph{probabilistically calibrated} if $Z_F$ is uniformly
    distributed on $[0,1]$.
    \item \label{def:marginal.calibration}
    The forecast $F$ is \emph{marginally calibrated} if
    $\E_{\Q}{F(y)} = \Q(Y \le y)$
    for all $y \in \R$.
  \end{enumerate}
\end{defi}

In part \ref{def:marginal.calibration} of Definition \ref{def:calibration.dispersion}
the left-hand side of the equation depends only on the distribution of the forecast, whereas the right-hand side depends only on the distribution of the observation. Marginal calibration therefore assesses whether the average forecast distribution is equal to the marginal distribution of $Y$. If $F_{1,t}$ is marginally cross-calibrated with respect to $F_{j,t}$, then, on average, the PIT $Z_{1,t}$ of $F_{1,t}$ behaves like a standard uniform random variable when considered in view of $F_{j,t}$. Intuitively, this means that $F_{1,t}$ has enough information about $F_{j,t}$ and the observation $Y_{t+1}$ to disguise itself as uniform on average when viewed through the eyes of $F_{j,t}$. 

Probabilistic cross-calibration means that the PIT $Z_{1,t}$ of $F_{1,t}$ is uniformly distributed no matter what the other forecasters predict. In contrast, probabilistic calibration of $F_i$ means that $Z_{F_i}$ is uniformly distributed on average over all possible predictions of the other forecasters, which is a weaker notion.

\begin{rem}
\citet{GneitingRanjan2013} also formalize the concept of dispersion in their Definition 2.6 in terms of the variance of $Z_F$. It is possible to define a notion of cross-dispersion for multiple forecasters considering the conditional variance of $Z_{1,t}$ given $F_{i_1,t}$, \dots, $F_{i_m,t}$. However, we feel that formulating dispersion in terms of the variance of $Z_F$ is not as natural as it seems at first sight. If $F$ is probabilistically calibrated then $Z_F$ is uniformly distributed on $[0,1]$, therefore its variance is $1/12$ and $F$ is called neutrally dispersed. In this case, $\Phi^{-1}(Z_F)$, would have a standard normal distribution, where $\Phi^{-1}$ denotes the quantile function of the standard normal distribution. It is equally intuitive to define dispersion in terms of the variance of $\Phi^{-1}(Z_F)$ with over- and underdispersion if this variance is smaller or larger than one, respectively. If a random variable $X$ with values in $[0,1]$ has variance $1/12$,  generally, it does not follow that $\Phi^{-1}(X)$ has unit variance. For example, let $X$ be a beta distributed random variable with parameters $\alpha=1$ and $\beta = (\sqrt{33}-5)/2$. Then $X$ has variance $1/12\approx 0.083$. The variance of $\Phi^{-1}(X)$ is approximately $1.92 \not= 1$. Therefore, it may well be that a forecast is neutrally dispersed with respect to $Z_F$ but over- or underdispersed with respect to $\Phi^{-1}(Z_F)$. Due to this ambiguity, we do not consider the concept of (cross-)dispersion in this manuscript. 
\end{rem}

The following theorem formally connects Definitions \ref{def:cross.calibration} and \ref{def:calibration.dispersion} showing that the former is indeed a generalization of the latter.

\begin{thm}
  \label{thm:cross.individual}
Consider forecasters $F_{1,t},\dots,F_{k,t}$ in a prediction space for serial dependence.
\begin{enumerate}
\item
The forecast $F_{1,t}$ is marginally cross-calibrated with respect to itself, if and only if $F_{1,t}$ is marginally calibrated.
\item 
If $F_{1,t}$ is cross-calibrated with respect to $F_{i_1,t}, \ldots, F_{i_m,t}$, then $F_{1,t}$ is
cross-calibrated with respect to any subset of $\{i_1,\dots,i_m\}$. In particular, $F_{1,t}$ is cross-calibrated with respect to the empty set $\emptyset$, that is, probabilistically calibrated.
\item If $F_{1,t}$ is cross-calibrated with respect to $F_{2,t}$, then it is also marginally cross-calibrated with respect to $F_{2,t}$. 
\end{enumerate}
\end{thm}
\begin{proof}
To show the first claim, observe that we have for all $y \in \R$,
\begin{align*}
  \E_{\Q}{\one\{F_{1,t}^{-1}(Z_{1,t}) \le y\}} &= \Q[F_{1,t}(Y_{t+1}-) + V\{F_{1,t}(Y_{t+1}) -
    F_{1,t}(Y_{t+1}-)\} \le F_{1,t}(y)]\\
  &= \Q\{F_{1,t}(Y_{t+1}) \le F_{1,t}(y)\}\\
  & = \Q(Y_{t+1} \le y).
\end{align*}
The second equality holds, because 
\begin{displaymath}
Z_{1,t} = F_{1,t}(Y_{t+1}-) + V\{F_{1,t}(Y_{t+1}) -
F_{1,t}(Y_{t+1}-)\} \in [F_{1,t}(Y_{t+1}-),F_{1,t}(Y_{t+1})],
\end{displaymath}
 where the interval consists of the point $F_{1,t}(Y_{t+1}-) = F_{1,t}(Y_{t+1})$ if $F_{1,t}$ is
 continuous at the point $Y_{t+1}$, and $Z_{1,t}\in (F_{1,t}(Y_{t+1}-),F_{1,t}(Y_{t+1}))$
 almost surely, otherwise. Furthermore, $F_{1,t}(y) \le F_{1,t}(Y_{t+1}-)$ or $F_{1,t}(y) \ge F_{1,t}(Y_{t+1})$. Let $J \subset \{i_1,\dots,i_m\}$.
The second claim follows because, for $y \in (0,1)$,
\begin{align*}
\Q(Z_{1,t} \le y \mid F_{i,t}, i \in J, \cT_t) &= \E_{\Q}\{\Q(Z_{1,t} \le y|\, F_{i_1}, \ldots, F_{i_m}, \cT_t)\mid F_{i,t}, i \in J, \cT_t\}\\
& = \E_{\Q}(y\mid F_{i,t}, i \in J, \cT_t) = y
\end{align*}
by the definition of cross-calibration. The last claim holds because \[\E_{\Q}\one\{F_{2,t}^{-1}(Z_{1,t}) \le y\} = \E_{\Q}\Q\{Z_{1,t} \le F_{2,t}(y)\mid F_{2,t},\cT_t\} = \E_{\Q}F_{2,t}(y).\]
\end{proof}

It is possible that a forecaster is marginally calibrated but not probabilistically calibrated; see \citet[Example 2.4]{GneitingRanjan2013} which we take up below in Example \ref{ex:marginal.calibration} to illustrate cross-calibration. Conversely, the last claim of Theorem \ref{thm:cross.individual} shows that marginal cross-calibration with respect to a different forecaster is a necessary condition for cross-calibration.

\citet{Tsyplakov2011,Tsyplakov2013} introduced a slightly more restrictive notion than an ideal forecaster, which is an \emph{auto-calibrated} forecaster, that is, it fulfils $\mathcal{L}(Y \mid F) = F$, almost surely, in the one-period prediction space setting. Generally, an auto-calibrated forecaster is ideal with respect to $\sigma(F)$, which is the $\sigma$-algebra generated by $F$. \citet{GneitingRanjan2013} contend that it is unlikely that empirical test of auto-calibration are feasible, except for very special circumstances such as forecasts for binary random variables. In cases where forecasters are restricted to specific classes of distributions \citet{HeldRufibachETAL2010} have taken on the challenge to derive statistical tests for ideal forecasters in the sense of auto-calibration based on a score regression approach; for earlier work in this direction see \citet{Hamill2001,MasonGalpinETAL2007}. In Section \ref{sec:cross.ideal.test} we show that it is possible to extend the score regression approach of \citet{HeldRufibachETAL2010} to test for cross-calibrated forecasters, that is, for cross-ideal forecasters with respect to $\sigma(F_1),\dots,\sigma(F_k)$; compare Proposition \ref{prop:pccci}.

In this paper, we challenge the statement of \citet{GneitingRanjan2013} by proposing two powerful tests for cross-calibration under very general assumptions that are justified even under serial dependence; see Sections \ref{sec:condex} and \ref{sec:prob.cross.cal.test}. Note that the following Proposition \ref{prop:pccci} shows that auto-calibration is in fact a special case of cross-calibration. 

\begin{prop}\label{prop:pccci}
Consider forecasters $F_{1,t},\dots,F_{k,t}$ in a prediction space for serial dependence. Let $\{i_1,\dots,i_m\} \subset \{1,\dots,k\}$. Then, the following are equivalent:
\begin{enumerate}
\item The forecaster $F_{1,t}$ is cross-calibrated with respect to $F_{i_1,t},\dots,F_{i_m,t}$.
\item For all $z \in [0,1)$, conditional on $F_{i_1,t},\dots,F_{i_m,t},\cT_t$, the random variable $\one\{Z_{1,t} \le z\}$ is Bernoulli distributed with parameter $z$. 
\end{enumerate}
If $1 \in \{i_1,\dots,i_m\}$, then part one and two are equivalent to $F_{1,t}$ being cross-ideal with respect to $\sigma(F_{i_1,t}),\dots,\sigma(F_{i_m,t})$.
\end{prop}

\begin{proof} The equivalence of parts one and two is immediate from the definition of cross-calibration. Suppose now that $1 \in \{i_1,\dots,i_m\}$. For all $y \in \R$, we obtain
\begin{multline*}
\Q(Y_{t+1} \le y |\, F_{i_1,t},\dots,F_{i_m,t}, \cT_t) = \Q\{F_{1,t}^{-1}(Z_{1,t})\le y|\,F_{i_1,t},\dots,F_{i_m,t}, \cT_t\} \\= \Q\{Z_{1,t} \le F_{1,t}(y)|\,F_{i_1,t},\dots,F_{i_m,t}, \cT_t\} = F_{1,t}(y),
\end{multline*}
which shows that last claim.
\end{proof}

We conclude this section with the announced generalization of the result of \citet{DieboldGuntherETAL1998} on uniformity and independence of PIT values in a prediction space for serial dependence.

\begin{thm}
  \label{thm:serial.dependence} Suppose we are in the prediction space setting for serial dependence. Let $\{i_1,\dots,i_m\}\subset \{1,\dots,k\}$ and assume that $F_{1,t} = \mathcal{L}(Y_{t+1} \vert \A_{i_1,t}, \ldots, \A_{i_m,t}, \mathcal{T}_t)$ for all $t \in \N$. Then, for all $l \in \N_0$, we have
\begin{displaymath}
  \mathcal{L}(Z_{1,t}, \ldots, Z_{1,t+l} \vert \A_{i_1,t+l}, \ldots, \A_{i_m,t+l}) = \mathcal{U}([0,1])^{\otimes (l+1)}, \quad \text{almost surely,}
\end{displaymath}
for all $t \in \N$. Here, $\mathcal{U}([0,1])^{\otimes (l+1)}$ denotes the distribution of $l+1$ independent standard uniform random variables. 
\end{thm}
\begin{proof} We define $\cB_t:= \sigma(\A_{i_1,t}, \dots, \A_{i_m,t})$. For $u = (u_0,\dots,u_l) \in (0,1)^{l+1}$, we obtain
\begin{align*}
\mathbb{E}&\{\one(Z_{1,t} \le u_0)\cdots\one(Z_{1,t+l} \le u_l)\mid \cB_{t+l}\}\\
&=\mathbb{E}\big[\one(Z_{1,t} \le u_0)\cdots\one(Z_{1,t+l-1} \le u_{l-1})\mathbb{E}\{\one(Z_{1,t+l} \le u_l)\mid \cB_{t+l},\mathcal{T}_{t+l}\}\mid \cB_{t+l}\big]\\
&=\mathbb{E}\big[\one(Z_{1,t} \le u_0)\cdots\one(Z_{1,t+l-2} \le u_{l-2})\E\{\one(Z_{1,t+l-1} \le u_{l-1})\mid \cB_{t+l},\cT_{t+l-1}\}\mid \cB_{t+l}\big]u_l\\
&=\mathbb{E}\big[\one(Z_{1,t} \le u_0)\cdots\one(Z_{1,t+l-2} \le u_{l-2})\E\{\one(Z_{1,t+l-1} \le u_{l-1})\mid \cB_{t+l-1},\cT_{t+l-1}\}\mid \cB_{t+l}\big]u_l\\
&=\mathbb{E}\big[\one(Z_{1,t} \le u_0)\cdots\one(Z_{1,t+l-2} \le u_{l-2})\mid \cB_{t+l-1}\big]u_{l-1}u_l = \dots = u_0\cdots u_l,
\end{align*}
where we used condition \eqref{eq:serialass} to obtain the third equality, and then proceeded iteratively.
\end{proof}

\begin{rem}\label{rem:qstep}
If we consider $q$-step ahead forecasts for some $q \ge 2$, then the above result continues to hold for all vectors of the form
\[
(Z_{1,t},Z_{1,t+q},\dots,Z_{1,t+mq}).
\]
However, there may be dependence amongst $(Z_{1,t},Z_{1,t+1},\dots,Z_{1,t+q-1})$, which complicates matters when testing for cross-calibration. This problem also arises in tests for uniformity and independence of PIT values as suggested by \citet{DieboldGuntherETAL1998}. Several approaches to deal with this issue have been suggested in the literature; see \citet{Knuppel2015} and references therein. In this paper, we restrict our attention to cross-calibration of one-period ahead forecasts but extensions to $q$-step ahead forecasts would certainly be of great interest.
\end{rem}

\section{Diagnostic plots for assessing cross-calibration}\label{sec:diagnostic}

\citet{GneitingBalabdaouETAL2007} suggest to assess marginal calibration based on a plot of the empirical analogue of the difference
\[
\E_{\Q}(F_t(y)) - \Q(Y_{t+1} \le y), \quad \text{for $y \in \R$.}
\]
Analogously, to assess marginal cross-calibration, the empirical version of 
\begin{equation}\label{eq:mar-cal-plot}
  \E_{\Q}{F_{j,t}(y)} - \E_{\Q}{\one\{F_{j,t}^{-1}(Z_{i,t}) \le y\}}, \quad \text{for $y \in \R$,}
\end{equation}
can be plotted. If the graph is not equal to zero everywhere one can deduce that $F_{t,i}$ is not marginally cross-calibrated with respect to $F_{j,t}$ and therefore also not cross-calibrated with respect to $\{j\}$ by Theorem \ref{thm:cross.individual}. If the graph is zero everywhere, then we have marginal cross-calibration. However, this does not necessarily imply that we have a cross-calibrated forecaster. 

Probabilistic calibration is often checked empirically by plotting a histogram of $Z_{i,t}$, the so-called PIT-histogram. Generally, it is not obvious how to check cross-calibration empirically. However, in many situations of practical interest it can be done by borrowing the idea of considering forecasting profiles as in the cross-calibration test of \citet{FeinbergStewart2008}. Suppose that the forecasters $F_1,\dots,F_k$ pick predictions from some parametric class of distributions $\mathcal{F} = \{F_{\lambda} \;|\; \lambda \in \Lambda\}$, where $\Lambda \subset \mathbb{R}^d$. Then we can identify each forecaster $F_{i,t}$ with the parameter $\lambda_{i,t}$ she predicts. We observe a sample $(F_{1,t},\dots,F_{k,t},Y_{t+1},V_t)$ for $1 \le t \le N$.  
Let $\Lambda_1,\dots,\Lambda_p$ be a partition of the parameter space. For a diagnostic plot showing if $F_{1,t}$, say, is cross-calibrated with respect to $\{i_1,\dots,i_m\}$, we can sort the observations into $pm$ bins according to the predicted values $(\lambda_{i_1,t},\dots,\lambda_{i_m,t})$. Then a PIT-histogram of $Z_{1,t}$ can be plotted for each bin. Clearly, the number of bins needs to be small in relation to the number of observations. 

We illustrate these diagnostic tools with two examples. The first one has been proposed by \citet[Examples 2.4]{GneitingRanjan2013}; see also \citet{GneitingBalabdaouETAL2007}.

\begin{ex}
\label{ex:marginal.calibration}
Let $\mu$ be standard normally distributed, which we denote by $\mu \sim \mathcal{N}(0,1)$. Conditional on $\mu$, the outcome is $Y\sim \mathcal{N}(\mu,1)$. Let $\tau$ take the values 1 or -1 with equal probability, independent of $Y$ and $\mu$. 
We consider four forecasters $F_1,\dots,F_4$ of different skill, whose properties are summarized in Table \ref{tab:cross.ideal}. 

It is clear that the perfect forecaster $F_1$ is cross-calibrated with respect to $F_1,F_2,F_3,F_4$. It is straight forward to check that the climatological forecaster $F_2$ is not cross-calibrated with respect to any of $F_1,F_3,F_4$ but with respect to itself. As $F_2$ is deterministic, this corresponds to the fact that $F_2$ is ideal with respect to the trivial $\sigma$-algebra. As the sign-reversed forecaster $F_4$ is not probabilistically calibrated it cannot be cross-calibrated. The cross-calibration of $F_3$ with respect to $F_1,F_2,F_4$ is shown in Appendix \ref{appendix:marg.cross.cal}. 
The statements about marginal cross-calibration in Table \ref{tab:cross.ideal} are consequences of Theorem \ref{thm:cross.individual}.


In Figure \ref{fig:marg.cross-calibration} the differences given at \eqref{eq:mar-cal-plot} are plotted. More precisely, the random variables are simulated 10'000 times and the mean is given. Recall, that all for all simulation examples we are using independent forecast-observation tuples for reasons of simplicity. In this example, it is easy to see that $F_1$ is superior to $F_2$ using the notion of marginal cross-calibration, which was not the case using only the calibration notions of \citet[Definition 2.6]{GneitingRanjan2013}; see \citet{GneitingBalabdaouETAL2007}.

As an example for checking cross-calibration empirically we note that all four forecaster are in the class of distribution functions $\mathcal{F} = \bigl\{F_\lambda \vert \lambda = (\mu,\sigma,\tau) \in \R \times (0,\infty) \times \{-1,0,1\}\bigr\}$ for $F_\lambda = \frac{1}{2} \{ \mathcal{N}(\mu,\sigma) + \mathcal{N}(\mu + \tau, \sigma)\}$. We plotted the PIT-histograms of $Z_{2}$ and $Z_{3}$ conditional on the four bins $\mu \in I_j$ for $1 \leq j \leq 4$ with $I_1 = (- \infty,-0.67)$, $I_2 = [-0.67,0)$, $I_3 = [0,0.67)$, $I_4 = [0.67,\infty)$. The PIT-histograms in Figure \ref{fig:ProbCrossCalGn} confirm that $F_3$ is probabilistic cross-calibrated with respect to $F_1,F_2,F_4$. On the other hand, $F_2$ is clearly not probabilistic cross-calibrated with respect to any set of the other forecasters.
\end{ex}

\begin{table}
\begin{center}
\begin{tabular*}{0.77\textwidth}{|@{\extracolsep{\fill} }l|l|l|} \hline
  Forecaster     & Predictive distribution  & Information set  \\ \hline
  Perfect        & $F_1 = \mathcal{N}(\mu,1)$ & $\mathcal{A}_1 = \sigma(\mu)$ \\
  Climatological & $F_2 = \mathcal{N}(0,2)$ & $\mathcal{A}_2 = \{\emptyset,\Omega\}$ \\
  Unfocused      & $F_3 = \frac{1}{2}\{\mathcal{N}(\mu,1) + \mathcal{N}(\mu + \tau,1)\}$ &
                   $\mathcal{A}_3 = \sigma(\mu, \tau)$ \\
  Sign-reversed  & $F_4 = \mathcal{N}(-\mu,1)$ & $\mathcal{A}_4 = \sigma(\mu)$ \\ \hline
\end{tabular*}\vspace{1pt}
\begin{tabular*}{0.77\textwidth}{|@{\extracolsep{\fill} }l|l|p{0.8cm}p{0.8cm}p{0.8cm}p{0.8cm}|}
\hline
Forecaster & Cross-calibration &\multicolumn{4}{l|}{Marginal cross-calibration wrt}\\
&  & $F_1$ & $F_2$ & $F_3$ & $F_4$ \\\hline
Perfect & wrt $F_1,F_2,F_3,F_4$ & yes & yes & yes & yes\\
Climatological & wrt $F_2$ & no & yes & no & no\\
Unfocused &  wrt $F_1,F_2,F_4$ & yes & yes & no & yes \\
Sign-reversed &  no & no & no & no & yes\\\hline
\end{tabular*}
\caption{Properties of the forecasters of \citet[Example 2.4]{GneitingRanjan2013}. Further details are given in Example \ref{ex:marginal.calibration}. Cross-calibration with respect to (wrt) $F_2$ is equivalent to cross-calibration with respect to $\emptyset$, that is, probabilistic calibration.\label{tab:cross.ideal}}
\end{center}
\end{table}

\begin{figure}
  \centering
  \includegraphics[scale=0.9]{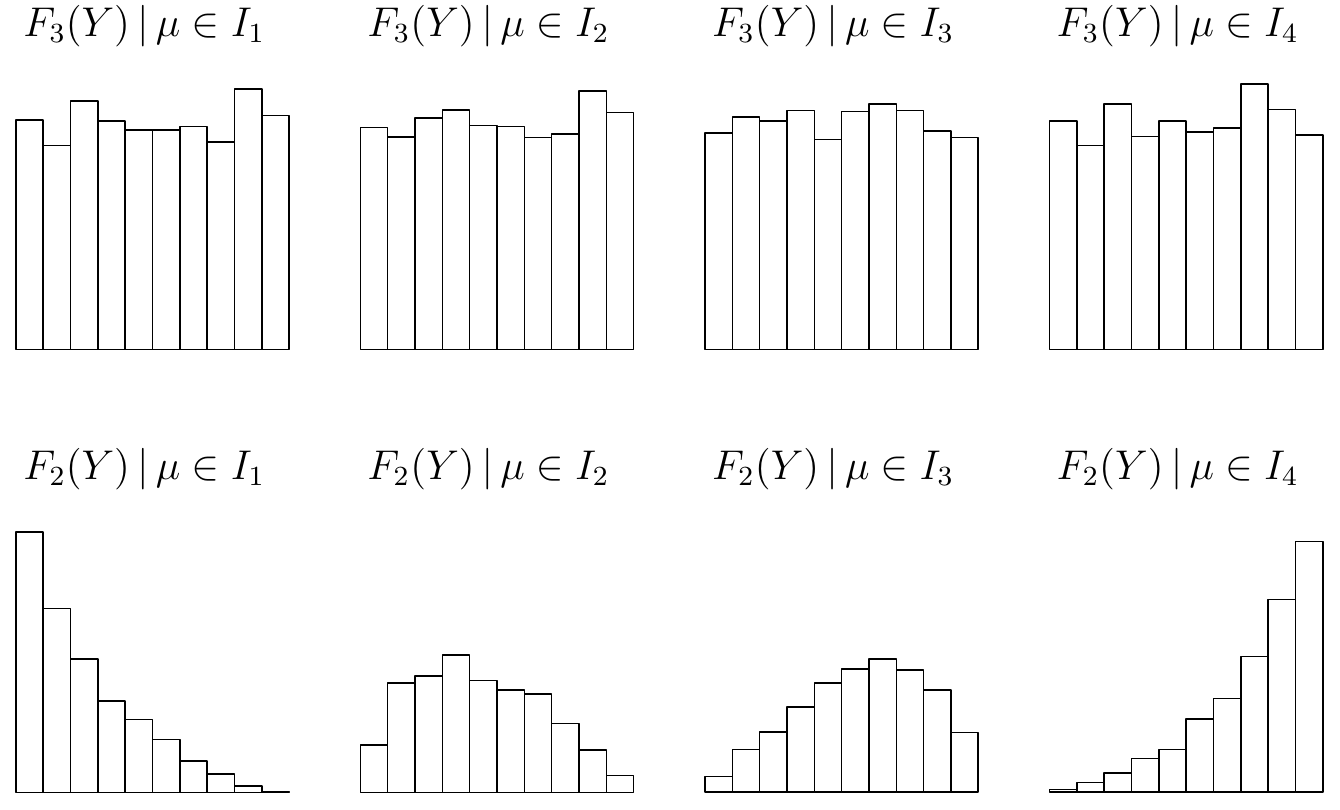}
  \caption{PIT-histogram plots of $F_3$ and $F_2$ conditional on the mean in the top and bottom row, respectively.\label{fig:ProbCrossCalGn}}
\end{figure}

\begin{figure}
  \centering
  \includegraphics[width=0.8\textwidth]{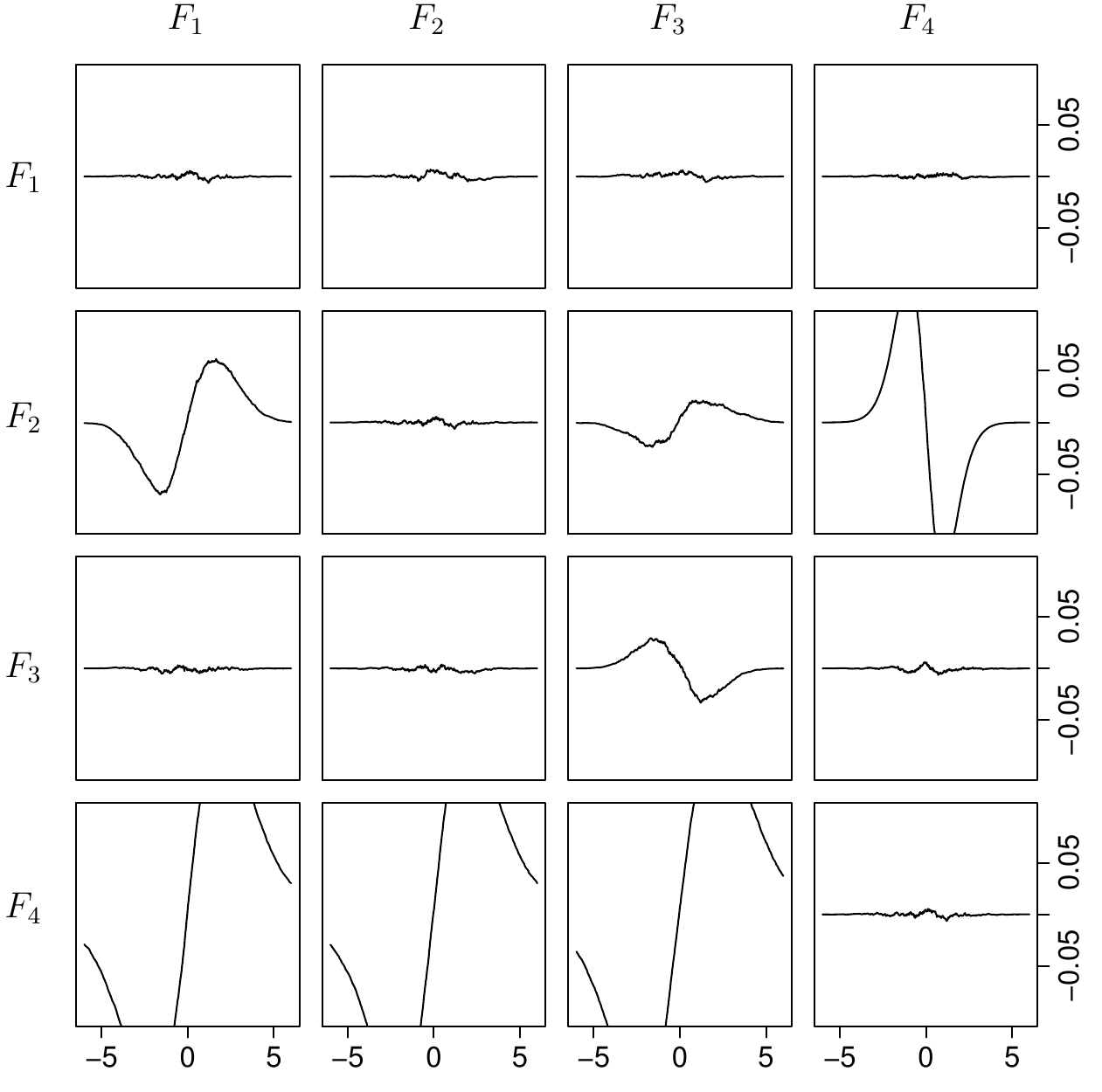}
  \caption{Marginal cross-calibration plots of the forecasters in Example \ref{ex:marginal.calibration} with 10'000 simulations. In the i-th row and j-th column the empirical version of Equation (\ref{eq:mar-cal-plot}) is plotted to assess whether $F_i$ is marginally cross-calibrated with respect to $F_j$ or not.\label{fig:marg.cross-calibration}}
\end{figure}

\begin{ex}[Example \ref{ex:cross-ideal} continued]\label{ex:2.5}
Coming back to the forecasters $F_1$ and $F_2$ of Example \ref{ex:cross-ideal}, Theorem \ref{thm:cross.individual} implies that $F_1$ is marginally cross-calibrated with respect to itself and with respect to $F_2$. Furthermore, $F_2$ is marginally cross-calibrated with respect to itself but $F_2$ is not marginally cross-calibrated with respect to $F_1$. Marginal cross-calibration plots for this scenario using 10'000 and 100'000 simulations are given in Figure \ref{fig:mar-cross-ideal}. In this example, the lack of marginal cross-calibration can only be detected for an unrealistically large number of observations. 

PIT-histograms for assessing cross-calibration of $F_1$ with respect to $F_2$  and of $F_2$  with respect to $F_1$ for 10'000 simulations are given in Figure \ref{fig:prob-cross-ideal}. The partition of the parameter space is chosen such that in each histogram there are around the same amount of observations. The lack of cross-calibration of $F_2$ with respect to $F_1$ is clearly detected. 

\end{ex}
\begin{figure}
	\centering
	\includegraphics{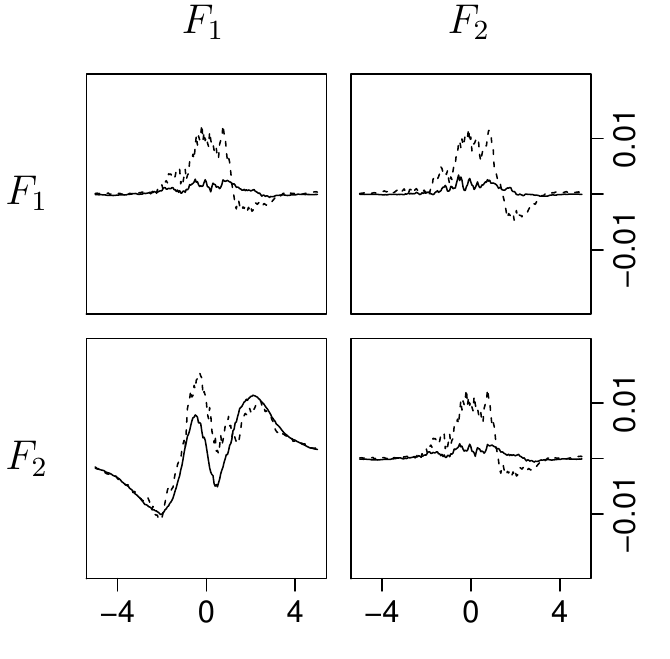}
	\caption{Marginal cross-calibration plots for the scenario in Example \ref{ex:cross-ideal} for 10'000 simulation indicated by the dashed lines and 100'000 simulations indicated by the continuous lines. The i-th row and j-th column corresponds to the empirical version of Equation (\ref{eq:mar-cal-plot}) in order to deduce whether $F_i$ is marginally cross-calibrated with respect to $F_j$ or not.\label{fig:mar-cross-ideal}}
\end{figure}
	\begin{figure}
	\centering
	\includegraphics[scale=0.9]{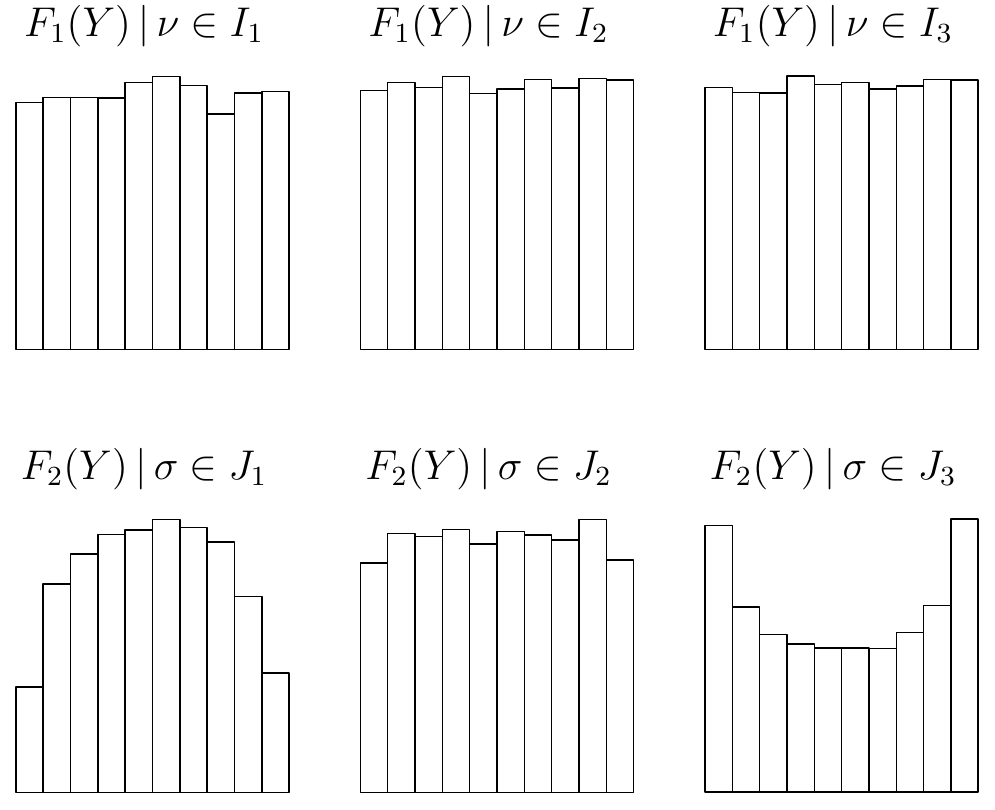}
	\caption{PIT-histogram plots of $F_1$ conditional on the predictive degree of freedom of $F_2$, where $I_1 = (5,10], I_2 = (10,15], I_3 = (15,20]$ in the top row, and PIT-histogram plots of $F_2$ conditional on the predictive standard deviation of $F_1$,  where $J_1 = [0,0.95], J_2 = (0.95,1.1], J_3 = (1.1,\infty]$ in the bottom row.\label{fig:prob-cross-ideal}}
\end{figure}

\section{Binary outcomes}
\label{sec:binary.outcomes}

In this section we consider the case, when the observation $Y$ only takes two values, zero and one. We interpret $Y = 1$ as a success and $Y = 0$ as a failure. A forecaster $F$ is
then represented by her predictive success probability $p$, such that the predictive CDF is
$F(y) = p \cdot \one\{y \ge 1\} + (1-p) \cdot \one\{y \ge 0\}$. We
identify $F$ with $p$, where $p$ is a random variable taking values in $[0,1]$.

In the case of an individual forecaster $F$ for a binary outcome it has been shown in
\citet[Theorem 2.11]{GneitingRanjan2013} that the notions of a probabilistically calibrated forecaster $F$ and an ideal forecaster relative to the $\sigma$-algebra generated by the predictive probability $p$ are equivalent. Furthermore, both notions coincide with the notion of \emph{conditional calibration}, that is $\mathbb{Q}(Y = 1|\,p) = p$. This result carries over to the notions of cross-calibration of multiple forecasters introduced in this paper. As the notions of calibration are essentially only concerned with one prediction period, we have chosen to present the results of this section in the one-period prediction space setting of Definition \ref{def:prediction.space} for simplicity.
\begin{thm}
  \label{thm:binary.outcomes}
Consider the one-period prediction space setting with binary outcome $Y$ and forecasts
$F_1,\dots,F_k$ represented by their predictive success probabilities
$p_1,\dots,p_k$, respectively. Then the following statements are equivalent:
  \begin{enumerate}
    \item \label{thm:bo.prob} The forecast $p_1$ is cross-calibrated
      with respect to $p_2,\dots, p_k$, that is $\mathcal{L}(Z_{p_1}|\,p_2,\dots, p_k)$ is standard uniform.
    \item \label{thm:bo.cond} The forecast $p_1$ is conditionally cross-calibrated
      with respect to $p_1,\dots,p_k$, that is $\Q(Y=1|\,p_1,\dots,p_k) = p_1$.
    \item \label{thm:bo.ideal} The forecast $p_1$ is ideal relative to $\sigma(p_1,\dots,p_k)$.
  \end{enumerate}
\end{thm}

The proof of Theorem \ref{thm:binary.outcomes} parallels the proof of \citet[Theorem 2.11]{GneitingRanjan2013}. The following lemma gives a formula for the density function of $Z_{p_1}$ conditional on $p_1 = x_1, \ldots, p_k = x_k$. 
\begin{lem}
  \label{lemma:density}
The density function of $Z_{p_1}$ conditional on $\mathbf{p} = \mathbf{x}$ is given by
\begin{displaymath}
  u(z|\, \mathbf{p}=\mathbf{x}) = \frac{q(\mathbf{x})}{x_1} \; \one(1-x_1 \le z \le 1) +
  \frac{1-q(\mathbf{x})}{1-x_1} \; \one(0 \le z < 1-x_1),
\end{displaymath}
where $\mathbf{p} = (p_1,\dots,p_k)$, $\mathbf{x} = (x_1,\dots,x_k) \in
[0,1]^{k}$, $x_1 \in (0,1)$ and $q(\mathbf{x}) = \Q(Y=1|\, \mathbf{p} = \mathbf{x})$.
\end{lem}

\begin{proof}[Proof of Lemma \ref{lemma:density}]
The PIT of $p_1$ is
\begin{displaymath}
  Z_{p_1} = \begin{cases} (1-p_1) + p_1 V, & \mbox{if }Y = 1,\\ (1-p_1)V, & \mbox{if }
    Y=0. \end{cases} 
\end{displaymath}
Let $0 \le z \le 1$, then
  \begin{align*}
    \Q(Z_{p_1} \le z\,|\, \mathbf{p}=\mathbf{x}) &= \Q\{(1-p_1) + p_1 V \le z, Y=1|\,
    \mathbf{p}=\mathbf{x}\}
    \\&\quad + \Q\{(1-p_1) V \le z,Y=0\,|\,\mathbf{p}=\mathbf{x}\}\\
    & = \Q\{(1-p_1) + p_1 V \le z\,|\, Y=1,\mathbf{p}=\mathbf{x}\} \Q(Y=1|\,
    \mathbf{p}=\mathbf{x})\\ 
    & \quad + \Q\{(1-p_1) V \le z\,|\, Y=0, \mathbf{p}=\mathbf{x}\}
    \Q(Y=0|\,\mathbf{p}=\mathbf{x})\\
    &= \frac{z + x_1 -1}{x_1}  q(\mathbf{x})  \one(1-x_1 \le z) + \{1-q(\mathbf{x})\} \one(1-x_1 \le z) \\
    &\quad + \frac{z}{1-x_1}\{1-q(\mathbf{x})\}  \one\{1-x_1 >z\}\\
    &= \frac{1-q(\mathbf{x})}{1-x_1}  z  \one(1-x_1 > z) + \left\{1-
      \frac{q(\mathbf{x})}{x_1} + \frac{q(\mathbf{x})}{x_1} z \right\}
    \one(1-x_1 \le z).
  \end{align*}
Differentiation yields the claim.
\end{proof}

\begin{proof}[Proof of Theorem \ref{thm:binary.outcomes}]
  It is easy to see that part two is equivalent to part three. By Theorem \ref{thm:cross.individual}, part three implies part one. The remaining task is to prove that part one implies part two. Let $\mathbb{H} = \mathbf{p}(\Q)$ be the marginal law of the random vector $\mathbf{p}$ under $\Q$. Recall that $q(\mathbf{x}) = \Q(Y=1|\, \mathbf{p} = \mathbf{x})$. If $\mathbb{H}(\{0\} \times [0,1]^{k-1}) > 0$, then $q(0,x_2,\ldots, x_k) = 0$ for all $ (x_2, \ldots, x_k) \in [0,1]^{k-1}$, because 
\begin{align*}
  \mathbb{H}(\{0\} \times [0,1]^{k-1}) &= \Q\{\mathbf{p}^{-1}(\{0\}
  \times [0,1]^{k-1})\} = \Q\{p_1^{-1}(0)\} = \Q(p_1 = 0),
\end{align*}
and furthermore,
\begin{align*}
  \Q(Z_{F_1} = 1 \,|\, p_2 = x_2,\ldots, p_k = x_k) &\ge \Q(Z_{F_1} = 1, Y = 1, p_1 = 0 \,|\, p_2 = x_2,\ldots,p_k = x_k)\\
  & = \Q(Y = 1, p_1 = 0 \,|\, p_2 = x_2,\ldots, p_k = x_k)\\
  & = \Q(Y = 1\,|\, p_1 = 0, p_2 = x_2,\ldots, p_k = x_k) \, \Q(p_1 = 0)\\
  &  = q(0,x_2,\ldots, x_k) \, \Q(p_1 = 0).
\end{align*}
We know that  $\Q(Z_{p_1} = 1|\, p_2 = x_2, \ldots, p_k = x_k) = 0$, because
$\mathcal{L}(Z_{p_1}|\, p_2,\ldots, p_k)$ is standard uniform. This implies that $q(0,x_2,
\ldots, x_k) = 0$. Similarly one can show that $\mathbb{H}(\{1\} \times [0,1]^{k-1}) >
0$ implies $q(1,x_2,\ldots,x_k) = 1$ for all $(x_2,\ldots,x_k) \in [0,1]^{k-1}$.

Using that $\mathcal{L}(Z_{p_1}|p_2 , \ldots, p_k)$ is a standard uniform
distribution and Lemma \ref{lemma:density}, we have for a.a.~$z \in (\,0,1)$, $\delta \in (\, 0, 1-z)$
\begin{align*}
  0 &= u(z+\delta|\,p_2=x_2, \ldots, p_k=x_k) - u(z|\, p_2=x_2, \ldots p_k=x_k)\\
    &= \int_{[0,1]} \left\{ u(z + \delta|\, \mathbf{p}=\mathbf{x})\, - u(z|\, \mathbf{p} =
      \mathbf{x}) \right\} d \mathbb{H}_1(x_1)\\
    &= \int_{[1-z-\delta,1-z)} \frac{q(\mathbf{x})-x_1}{x_1(1-x_1)}\, d \mathbb{H}_1(x_1),
\end{align*}
where $\mathbb{H}_1 = p_1(\Q)$ is the marginal law of $p_1$ under $\Q$. We define the
signed measure $\mu$ for a given $(x_2, \ldots, x_k) \in [0,1]^{k-1}$ as 
\begin{displaymath}
  \mu(A) = \int_A \frac{q(\mathbf{x})-x_1}{x_1(1-x_1)}\, d \mathbb{H}_1(x_1),
\end{displaymath}
for all Borel sets $A \subset [a,b]$, where $0 < a < b < 1$. For $[c,d) \subset [a,b]$
we have shown before, that
\begin{displaymath}
  \mu([c,d)) = \int_{[c,d)} \frac{q(\mathbf{x})-x_1}{x_1(1-x_1)} \, d
  \mathbb{H}_1(x_1) = 0.
\end{displaymath}
Therefore, $\mu(B) = 0$ for all $B \in \mathcal{B}([a,b])$. In particular, $\{x_1 \in [a,b]|\, q(\mathbf{x})>x_1\}$ and $\{x_1 \in [a,b]|\, q(\mathbf{x})<
 x_1\}$ are $\mathbb{H}_1$ null sets and we have $q(\mathbf{x}) = x_1$  $\mathbb{H}_1$-a.s., hence,
\begin{align*}
   p_1^{-1}\{q(\mathbf{x})=x_1)\} &= \{\omega: q\{p_1(\omega),x_2, \ldots, x_k\} =
   p_1(\omega)\}\\ &= \{ \Q(Y = 1 \, | \, p_1, p_2 = x_2, \ldots p_k = x_k) = p_1 \}
\end{align*}
has $\Q$-measure 1 for all $(x_2, \ldots, x_k) \in [0,1]^{k-1}$. Therefore, $\Q(Y = 1|\, \mathbf{p}) = p_1$ $\Q$-a.s..
\end{proof}

The cross-calibration notion of \citet{FeinbergStewart2008} is analogous to our notion of cross-calibration with respect to $\{1,\dots,k\}$ which is equivalent to cross-ideal forecasters for binary events. Theorem \ref{thm:binary.outcomes} shows that both notions coincide with cross-calibration of $p_1$ with respect to $\{2,\dots,k\}$ which is a priori a weaker requirement. As noted by \citet{GneitingRanjan2013} the fact that probabilistically cali\-brated forecasters are automatically ideal clarifies the relation between PIT-histograms and  \emph{calibration curves} which are the diagnostic tool frequently used for assessing calibration of binary predictions \citep{Dawid1986,MurphyWinkler1992,RanjanGneiting2010}. As described in Section \ref{sec:diagnostic}, cross-calibration can be assessed with conditional PIT-histograms. Analogously, in the case of binary forecasts, conditional calibration curves can be considered. 

\section{Tests for assessing cross-calibration}
\label{sec:tests}

In this section we consider statistical tests for cross-calibration. The tests in Section \ref{sec:condex} are based on the idea of conditional exceedance probabilities \citep{MasonGalpinETAL2007}, whereas the tests in Section \ref{sec:prob.cross.cal.test} use a linear regression approach. Finally, the score regression approach by \citet{HeldRufibachETAL2010} to test for ideal forecasters is reviewed and extended to a test for cross-ideal forecasters in Section \ref{sec:cross.ideal.test}.

\subsection{Conditional exceedance probabilities}\label{sec:condex}

Suppose we have observations $F_{1,t},\dots,F_{k,t}$ and $Y_{t+1}$, $1 \le t \le N$ in a prediction space for serial dependence. We would like to test the null hypothesis that $F_{1,t}$ is cross-calibrated with respect to $J \subset \{1,\dots,k\}$. For $z \in [0,1)$, we define $B_{z,t}:= \one\{Z_{1,t} \le z\}$. Under the null hypothesis, using Proposition \ref{prop:pccci} and Theorem \ref{thm:serial.dependence}, conditional on $F_{i,t}$ for $i \in J$, the random variables $B_{z,1},\dots,B_{z,N}$ are independent Bernoulli random variables with success probability $z$. We stipulate the logistic regression models
\begin{equation}\label{eq:logitmodel}
\logit[\Q\{B_{z,t} = 1|\, F_{i,t}^{-1}(z), i \in J\}] = \beta_{0,z} + \sum_{i \in J} \beta_{i,z} F_{i,t}^{-1}(z)
\end{equation}
for each $z \in [0,1)$, where $\logit(p) = \log\{p/(1-p)\}$ is the logistic function. Using \eqref{eq:logitmodel}, the null hypothesis is
\begin{equation}\label{eq:globalnull2}
H_0 \colon\quad \beta_{0,z} = \logit(z), \; \beta_{i,z} = 0,\; i \in J, \quad \text{for all $z \in [0,1)$}.
\end{equation}
For each $z \in [0,1)$, we suggest to test the pointwise hypothesis 
\begin{equation} \label{eq:pointwisenull}
H_0(z)\colon \quad \beta_{0,z} = \logit(z), \; \beta_{i,z} = 0,\; i \in J,
\end{equation}
by a likelihood ratio test yielding a $p$-value $\pi(z)$. 

More precisely, the covariate vector $\mathbf{x}_{z,t}$ has one as the first entry and then $F_{i,t}^{-1}(z)$, $i \in J$ and the parameter vector $\boldsymbol{\beta}_z$ has entries $\beta_{0,z}$, $\beta_{i,z}$, $i \in J$. For values of $z$ close to zero or one, we frequently encounter the phenomenon of separation, that is, the likelihood converges, but at least one parameter value is infinite. Therefore, we have chosen to use the method of \citet{Firth1993}, which always yields finite parameter estimates; see \citet{HeinzeSchemper2002}. That is, we fit the parameters $\beta_{0,z}$, $\beta_{i,z}$, $i \in J$ by maximizing the penalized log-likelihood function
\[
\ell_p(\boldsymbol{\beta}_z) = \ell(\boldsymbol{\beta}_z) + \frac{1}{2}\log |I(\boldsymbol{\beta}_z)|,
\]
where
\[
\ell(\boldsymbol{\beta}_z) = \sum_{t=1}^N B_{z,t} \mathbf{x}_{z,t}^{\top} \boldsymbol{\beta}_z - \sum_{t=1}^N \log\{1 + \exp(\mathbf{x}_{z,t}^{\top}\boldsymbol{\beta}_z)\},
\]
and $|I(\boldsymbol{\beta}_z)|$ is the determinant of the Fisher information matrix. 
We denote the estimated parameter vector by $\boldsymbol{\hat\beta}_z$ with entries $\hat{\beta}_{0,z}$, $\hat{\beta}_{i,z}$, $i \in J$. For $N$ large enough, the test statistic
\[
T_z = -2 \{\ell_p(\boldsymbol{\hat\beta}_z) - \ell_p(\boldsymbol{\gamma}_z)\}
\]
has a $\chi^2$-distribution with $1+|J|$ degrees of freedom, where $|J|$ denotes the cardinality of $|J|$, and $\boldsymbol{\gamma}_z = (\logit(z),0,\dots,0)^{\top}$. We define the $p$-value $\pi(z) = 1 - \chi^2_{1+|J|}(T_z)$, where $\chi^2_{1+|J|}$ denotes the cumulative distribution function of a $\chi^2$ random variable with $1+|J|$ degrees of freedom. For the simulation studies below and the data analysis in Section \ref{sec:data} we have used the R-package of \citet{HeinzePlonerETAL2013} to calculate $T_z$.

In order to draw conclusions about the global null hypothesis $H_0$ at \eqref{eq:globalnull2} from the pointwise $p$-values $\pi(z)$, we adjust them for multiple testing. We follow the approach of \citet{CoxLee2008} to use the method of \citet[Chapter 2]{WestfallYoung1993} for functional data to compute adjusted $p$-values $r(z)$; see also \citet{MeinshausMaathuisETAL2011}. 

Let $0 < z_1 <\dots < z_M < 1$. Under the null hypothesis of cross-calibration, it is possible to simulate a vector of $p$-values $(\pi^*(z_1),\dots,\pi^*(z_M))$ with the same distribution as $(\pi(z_1),\dots,\pi(z_M))$ conditional on $F_{i,t}^{-1}(z_m)$, $i \in J$, $1 \le t \le N$, $1\le m \le M$, as follows. Let $U_1,\dots,U_N$ be iid standard uniform random variables. For $1 \le m \le M$, define $B_{z_m,t}^* = \one(U_t \le z_m)$, and let $\pi^*(z_m)$ be the $p$-value from the pointwise likelihood ratio test for the simulated data vector $(B_{z_m,t}^*)_{1 \le t \le N}$ and covariates $(\mathbf{x}_{z_m,t})_{1 \le t \le N}$ as before.

The adjusted $p$-values can now be obtained as follows. Let $\sigma$ be the permutation of $\{1,\dots,M\}$ such that $\pi\{z_{\sigma(1)}\}\le \dots \le \pi\{z_{\sigma(M)}\}$. This permutation $\sigma$ remains unchanged in the following procedure. 
For a simulated vector of $p$-values $(\pi^*(z_1),\dots,\pi^*(z_M))$, we define $q_{m}^* = \min\{\pi^*\{z_{\sigma(s)}\}\colon s \ge m\}$. Repeating this procedure $L$ times, we obtain an array $(q_{m,l}^*)_{1\le m\le M,1 \le l\le L}$ and define the adjusted $p$-values $r_{1},\dots,r_{M}$ corresponding to  $z_1,\dots,z_M$ as
\[
r_m = \frac{1}{L}\sum_{l=1}^L \one\{q_{\sigma^{-1}(m),l}^* \le \pi(z_m)\},\quad 1 \le m \le M.
\]
The global null hypothesis $H_0$ at \eqref{eq:globalnull2} can be rejected at level $\alpha\in (0,1)$ if $\min\{r_m \colon 1 \le m \le M\} \le \alpha$. Furthermore, the adjusted $p$-values allow to draw conclusions for which values of $z_m \in (0,1)$ miscalibration occurs. For example, a prediction method may perform satisfactory for the left tail of the distribution, that is, for $z$ close to zero, the adjusted $p$-values are large, whereas it fails to capture the right tail and hence for $z$ close to one, the adjusted $p$-values are small. We call this test the \emph{CEP test with respect to $J$}. 

\begin{rem}
It is important to note that the adjusted $p$-values $r_m$ remain the same, if the pointwise $p$-values $\pi(z)$ are transformed with a strictly monotone transformation. Therefore, even if the $\pi(z)$ are only asymptotic $p$-values, the adjusted $p$-values $r_m$ will control the familywise error rate at the desired level $\alpha$ even for finite samples (for large numbers $L$ of bootstrap replications); see \citet[Chapter 2]{WestfallYoung1993} and \citet{CoxLee2008}. It is nevertheless important which test statistic to choose for the pointwise tests as the power of the overall test will crucially depend on the power of the pointwise tests. 
\end{rem}

\begin{ex}[Example \ref{ex:marginal.calibration} continued]\label{ex:CEP-GN}
We consider the forecasters $F_1,\dots,F_4$ of Example \ref{ex:marginal.calibration}; see Table \ref{tab:cross.ideal}. For sample size $N=50$, we performed the CEP tests for $F_1,\dots,F_4$ with respect to all possible subsets of $F_1,\dots,F_4$ and calculated the Monte-Carlo power based on $10'000$ simulations. We used the gridpoints $z_m = \{1+ (18/19)m\}/20$, $0 \le m \le 19$. The number of bootstrap replications for calculating the adjusted $p$-values is set to $L=500$. For data examples, $L$ should be much larger. However, for analysing the performance of the resampling based $p$-values, it is more important to run a large number of simulations than to have a large bootstrap sample for each of them; see \citet{WestfallYoung1993} for a more detailed discussion. The results are given in Table \ref{tab:CEP-GN}. 

Conditioning on $F_2$ corresponds to conditioning on the trivial $\sigma$-algebra, therefore testing conditional on $F_1,F_2,F_3$ is the same as testing conditional on $F_1,F_3$, for example. Hence, Table \ref{tab:CEP-GN} contains all interesting subsets of $F_1,\dots,F_4$ and the column entitled `$F_2$' corresponds to a test for probabilistic calibration. The test performs well, even for the small sample size $N=50$. Generally, the power of the test appears to increase, the more information is used. For example, the test has difficulty to detect that $F_3$ is not ideal with respect to itself but it performs well for rejecting the null hypothesis that $F_3$ is cross-ideal with respect to $F_1,F_3$, $F_3,F_4$ or $F_1,F_3,F_4$. 
\end{ex}

\begin{table}
\centering
\begin{tabular}{|c|cccccccc|}
\hline
wrt & $F_1$ & $F_2$ & $F_3$ & $F_4$ & $F_1,F_3$ & $F_1,F_4$ & $F_3,F_4$ & $F_1,F_3,F_4$\\\hline
$F_1$ & 0.056 & 0.051 & 0.055 & 0.055 &  0.050 & 0.051 &  0.050 & 0.048\\
$F_2$ & 0.997 & 0.051 & 0.979 & 0.997 &  0.994 & 0.990 &  0.993 & 0.977\\
$F_3$ & 0.052 & 0.052 & 0.168 & 0.051 &  0.635 & 0.051 &  0.634 & 0.582\\
$F_4$ & 1.000 & 1.000 & 1.000 & 1.000 &  1.000 & 1.000 &  1.000 & 1.000\\\hline
\end{tabular}
\caption{Monte-Carlo power for the  CEP tests at significance level $\alpha=0.05$. Details of the simulation study are given in Example \ref{ex:CEP-GN}.\label{tab:CEP-GN}}
\end{table}

\begin{ex}[Example \ref{ex:cross-ideal} continued]\label{ex:CEP-CI}
We applied the CEP tests to data simulated from the prediction space described in Example \ref{ex:cross-ideal}. We used the same grid and other parameters as in the previous example, except that we considered two different sample sizes $N=50$ and $N=200$. The results from $10'000$ simulations can be seen in Table \ref{tab:CEP-CI}. Here, the power for sample size $N=50$ is only small. Fortunately, it appears to increase rapidly with sample size and is satisfactory for $N=200$.
\end{ex}

\begin{table}
\begin{center}
\begin{tabular}{|c|ccc|}
\hline
$N=50$ & $F_1$ & $F_2$ & $F_1, F_2$ \\\hline
$F_1$ & 0.052 & 0.053 & 0.050\\
$F_2$ & 0.156 & 0.051 & 0.139\\\hline
\end{tabular}
\quad
\begin{tabular}{|c|ccc|}
\hline
$N=200$ & $F_1$ & $F_2$ & $F_1, F_2$ \\\hline
$F_1$ &  0.050 & 0.052 & 0.051\\
$F_2$ & 0.533 & 0.057 & 0.464\\
\hline
\end{tabular}
\end{center}
\caption{Monte-Carlo power for the CEP tests at significance level $\alpha=0.05$. Details of the simulation study are given in Example \ref{ex:CEP-CI}.\label{tab:CEP-CI}}
\end{table}

\subsection{Linear regression approach}
\label{sec:prob.cross.cal.test}

To formulate the linear regression approach (LRA) tests for cross-calibration, we restrict ourselves to a parametric class of cumulative distribution functions $\mathcal{F} = \{
F_\lambda|\, \lambda \in \Lambda\}$, where $\Lambda \subset \mathbb{R}^d$. Suppose we have $N$ forecast-observation tuples $(F_{1,t},\ldots,F_{k,t},Y_t,V_t)$, $1 \le t \le N$ in a prediction space for serial dependence such that $F_{i,t} \in \mathcal{F}$ for all $i$ and $t$. Each forecaster $F_{i,t}$ is then represented by its predictive parameter vector $\lambda_{i,t} = (\lambda_{i,t}^{(1)},\ldots, \lambda_{i,t}^{(d)})$ for $1 \le i \le k$. We want to test the hypothesis that $F_{1,t}$ is cross-calibrated with respect to some $J= \{i_1,\dots,i_m\} \subset \{1,\dots,k\}$, for $1 \le t \le N$.  
Proposition \ref{prop:pccci} and Theorem \ref{thm:serial.dependence} lead to the null hypothesis
\begin{displaymath}
  H_0 \colon \quad \mathcal{L}\{\Phi^{-1}(Z_{1,1}),\dots,\Phi^{-1}(Z_{1,N}) \vert \lambda_j, j \in J, 1\le t \le N\} = \cN_N(0,I_N), \quad \text{almost surely,}
\end{displaymath}
where $\Phi^{-1}$ denotes the quantile function of a standard normal distribution and $\cN_N(0,I_N)$ denotes a multivariate standard normal distribution. 
In order to test this hypothesis we perform an F-test based on linear regression. 
We consider the linear model
\begin{equation}\label{eq:linmod}
  \mathbf{Y} = \mathbf{D_J}\boldsymbol{\beta} + \boldsymbol{\epsilon},
\end{equation}
where
\begin{displaymath}
  \mathbf{Y} = (\Phi^{-1}(Z_{1,1}),\dots, \Phi^{-1}(Z_{1,N}))^T \in \R^N
\end{displaymath}
is the response vector, 
\begin{equation}\label{eq:design1}
  \mathbf{D_J} = \left(
  \begin{array}{*{7}c}
    1 & \lambda_{i_1,1}^{(1)} & \cdots & \lambda_{i_1,1}^{(d)} & \lambda_{i_2,1}^{(1)} & \cdots
    & \lambda_{i_m,1}^{(d)}\\
    1 & \lambda_{i_1,2}^{(1)} & \cdots & \lambda_{i_1,2}^{(d)} & \lambda_{i_2,2}^{(1)} & \cdots
    & \lambda_{i_m,2}^{(d)}\\
    \vdots & \vdots & \ddots & \vdots & \vdots & \ddots & \vdots\\
    1 & \lambda_{i_1,N}^{(1)} & \cdots & \lambda_{i_1,N}^{(d)} & \lambda_{i_2,N}^{(1)} & \cdots
    & \lambda_{i_m,N}^{(d)}\\
  \end{array}
  \right) \in \R^{N \times (1+dm)}
\end{equation}
is the design matrix,
\begin{displaymath}
  \boldsymbol{\beta} = (\beta_0,\beta_1, \ldots, \beta_{dk})^T \in \R^{1+ dm}
\end{displaymath}
is the parameter vector we would like to estimate and $\boldsymbol{\epsilon} \in \R^N$ is a random error vector, which is multivariate standard normal under the null hypothesis. 

In order to estimate $\boldsymbol{\beta}$ the method of least square is
used and we obtain the estimated parameter vector
\begin{displaymath}
  \boldsymbol{\hat{\beta}} = (\mathbf{D_J}^T \mathbf{D_J})^{-1} \mathbf{D_J}^T \mathbf{Y},
\end{displaymath}
the vector of fitted values
\begin{displaymath}
  \boldsymbol{\hat{Y}} = \mathbf{D_J} (\mathbf{D_J}^T \mathbf{D_J})^{-1} \mathbf{D_J}^T \mathbf{Y} =
  \mathbf{D_J} \boldsymbol{\hat{\beta}},
\end{displaymath}
and the residual vector $\boldsymbol{\hat{\epsilon}} = \mathbf{Y} - \mathbf{\hat{Y}}$.

Under the null hypothesis we have that
\begin{displaymath}
  \boldsymbol{\beta} = (0,0, \ldots, 0)^T \in \R^{1 + dm}\quad \mbox{and} \quad\boldsymbol{\epsilon} \sim \cN_N(0,I_N).
\end{displaymath}
To test the assumption that $\boldsymbol{\epsilon}$ is standard normal one can use a normality test such as the Anderson-Darling or Shapiro-Wilk \citep{AndersonDarling1954,ShapiroWilk1965,YapSim2011}. This yields a $p$-value $\pi_{normal}$.
To test the other assumption 
we consider the test statistic
\begin{displaymath}
  F_0 = \frac{\boldsymbol{\hat{\beta}}^T\!(\mathbf{D_J}^T \mathbf{D_J})\,
    \boldsymbol{\hat{\beta}}}{(1 + dm) \hat{\sigma}^2},
\end{displaymath}
where
\begin{displaymath}
  \hat{\sigma}^2 = \frac{\|\mathbf{Y} - \mathbf{\hat{Y}}\|^2}{N-(1 + dm)}
\end{displaymath}
is the unbiased variance estimator. The test statistic $F_0$ has a Fisher distribution with $1 + dm$ and $N - 1 - dm$ degrees of freedom; see for example \citet{MontgomerPeckETAL2001}. The $p$-value $\pi_F$ is then
\begin{displaymath}
  \pi_F = 1 - F_{1 + dm, N - 1 - dm}(F_0),
\end{displaymath}
where $F_{p,q}$ denotes the Fisher cumulative distribution function with $p$ and $q$ degrees of freedom. Combining these two tests by the method of Holm leads to the adjusted $p$-value 
\begin{displaymath}
  \pi_{\text{adjust}} = 2 \min\{ \pi_F,  \pi_{\text{normal}}\}.
\end{displaymath}

We need that $\operatorname{rank}(\mathbf{D_J}) = 1 + dm$, otherwise the regression analysis is not
possible. Therefore, any forecaster $F_{i,t}$, $i \in J$ has to predict for each parameter
at least two distinct values. Otherwise, we omit the parameter for this forecaster in the model and are still able to use the test, which we call the \emph{LRA test with respect to $J$}.

\begin{ex}[Example \ref{ex:marginal.calibration} continued]
  \label{ex:LRA.Gn}
Recall the forecasters $F_1,F_2,F_3$ and $F_4$ from Example \ref{ex:marginal.calibration}. Then all four forecasters are in the class of distribution functions $\mathcal{F} = \bigl\{F_\lambda \vert \lambda = (\mu,\sigma,\tau) \in \R \times (0,\infty) \times \{-1,0,1\}\bigr\}$ for $F_\lambda = \frac{1}{2} \{ \mathcal{N}(\mu,\sigma) + \mathcal{N}(\mu + \tau, \sigma)\}$. We apply the LRA test for all combinations of forecasters for sample sizes $N = 20$ and $N=50$. The Monte-Carlo powers of $\pi_{\text{adjust}}$ for $10'000$ simulations are given in Table \ref{tab:LRA.Gn}.
We only listed the Monte-Carlo powers with respect to individual forecasters, since we omit some parameters to have a design matrix with full rank. Therefore, testing cross-calibration with respect to $J \subset \{1,2,3,4\}$ leads to the same test as testing with respect to $F_3$ if $3 \in J$ and testing with respect to $F_1$ otherwise.  
For testing standard normality, we have used an Andersen-Darling test (with mean set to zero and variance set to one). In the cases of cross-calibration, the normality test never rejects the null hypothesis, which explains the conservative levels of around $0.025$ in these cases. 
The test is powerful even for the small sample sizes and it provides the expected results from the theoretical considerations; see Table \ref{tab:cross.ideal}. In particular, the LRA test detects well, that $F_3$ is not ideal with respect to itself contrary to the CEP test; compare Table \ref{tab:CEP-GN}. 
\begin{table}
  \centering
  \begin{tabular}{|c|cccc|}\hline
    $N = 20$& $F_1$ & $F_2$ & $F_3$ & $F_4$ \\ \hline
    $F_1$ & 0.024 & 0.026 & 0.025 & 0.025 \\
    $F_2$ & 0.884 & 0.025 & 0.825 & 0.884 \\
    $F_3$ & 0.024 & 0.025 & 0.238 & 0.024 \\
    $F_4$ & 1.000 & 0.880 & 0.999 & 1.000 \\ \hline
  \end{tabular}
\quad
  \begin{tabular}{|c|cccc|}\hline
    $N = 50$& $F_1$ & $F_2$ & $F_3$ & $F_4$ \\ \hline
    $F_1$ & 0.024 & 0.022 & 0.025 & 0.024 \\
    $F_2$ & 1.000 & 0.027 & 1.000 & 1.000 \\
    $F_3$ & 0.026 & 0.023 & 0.734 & 0.026 \\
    $F_4$ & 1.000 & 1.000 & 1.000 & 1.000 \\ \hline

  \end{tabular}
  \caption{Monte-Carlo power of the LRA test at significance level $\alpha=0.05$ for different sample sizes. Details are given in Example \ref{ex:LRA.Gn}. \label{tab:LRA.Gn}}
  
\end{table}
\end{ex}
\begin{ex}[Example \ref{ex:cross-ideal} continued]
\label{ex:prob-cross-ideal-test}
	Coming back to forecasters $F_1$ and $F_2$ from Example \ref{ex:cross-ideal} we perform the F-tests for different sample sizes $N$. The Monte-Carlo powers of the tests for $10'000$ simulations can be found in Table \ref{tab:prob-cross-ideal-test}. The Monte-Carlo powers are low even for large sample sizes, contrary to the results of the CEP tests; compare Table \ref{tab:CEP-CI}. We do not report the power of LRA test in this example because the Anderson-Darling test for standard normality almost never rejects the null hypothesis. 
\end{ex}

\begin{table}
\begin{center}
\begin{tabular}{|l | cccccc|} \hline
	$N$ & 20 & 50 & 100 & 200 & 1000 & 5000\\ \hline
	$F_1$ wrt $F_1$ & 0.051 & 0.049 & 0.048 & 0.049 & 0.053 & 0.047\\
	$F_1$ wrt $F_2$ & 0.053 & 0.05\phantom{0} & 0.045 & 0.05\phantom{0} & 0.05\phantom{0} & 0.05\phantom{0}\\
	$F_1$ wrt $F_1,F_2$ & 0.05\phantom{0} &  0.048 & 0.046 & 0.049 & 0.051 & 0.048\\
	$F_2$ wrt $F_1$ & 0.092 & 0.105 & 0.114 & 0.122 & 0.139 & 0.135\\ 
	$F_2$ wrt $F_2$ & 0.051 & 0.05\phantom{0} & 0.045 & 0.048 & 0.051 & 0.049 \\
	$F_2$ wrt $F_1,F_2$ & 0.081 & 0.092 & 0.097& 0.108 & 0.121 & 0.119\\ \hline
\end{tabular}
\end{center}
\caption{Monte-Carlo power for F-test for different sample sizes $N$ and $10'000$ simulations in Example \ref{ex:prob-cross-ideal-test}.\label{tab:prob-cross-ideal-test}}
\end{table}%

\subsection{Score regression approach}
\label{sec:cross.ideal.test}

\citet{HeldRufibachETAL2010} suggest a significance test for ideal forecasters based on scoring rules. They use the continuous ranked probability score (CRPS) \citep{GneitingRafteryETAL2005}
and the Dawid-Sebastiani score (DSS) \citep{DawidSebastian1999}. Their approach relies on independent forecast-observation tuples, and this restriction remains, when generalizing their approach to a test for cross-ideal forecasters. Therefore, throughout this section we work in a one-period prediction space.

First, we recall some preliminaries on the CRPS and the DSS. Let $F$ and $f$  denote the predictive CDF of a forecaster and the predictive density function,
respectively.  Let $\mu$ and $\sigma^2$ be the predictive mean and 
variance, respectively.\footnote{In the prediction space setting, the quantities $F$, $f$, $\mu$, and $\sigma$ are random $\mathcal{A}_0$-measurable quantities. For ease of presentation, in this section we treat them as as if they were deterministic, or, alternatively, one should consider all expectations and variances as conditional on the forecaster's information set $\mathcal{A}_0\subset \mathcal{A}$.} The observed value of $Y$ is
denoted by $y$. The CRPS is given by
\begin{displaymath}
  \emph{CRPS}(F,y) = \int_\infty^\infty \{F(x) - \one(y \le x)\}^2\, dx
\end{displaymath}
and the DSS by
\begin{displaymath}
  \emph{DSS}(F,y) = \frac{1}{2}\left\{\log(\sigma^2) + \tilde{y}^2\right\},
\end{displaymath}
where $\tilde{y} = (y- \mu)/\sigma$.
 For a forecaster predicting a normal distribution $F=\cN(\mu,\sigma^2)$, the CRPS turns out to be
\begin{displaymath}
  \emph{CRPS}(F,y) = \sigma \Big[\tilde{y}\{2\Phi(\tilde{y}) - 1\} + 2 \phi(\tilde{y}) - \frac{1}{\sqrt{\pi}}\Big],
\end{displaymath}
where $\Phi$ and $\phi$ are the CDF and the density of a standard normal distribution, respectively \citep{GneitingRaftery2007}. The CRPS is a strictly proper scoring rule relative to the class of probability measures with finite first moments; see \citet{GneitingRaftery2007} for details on proper and strictly proper scoring rules.

For a normal prediction the DSS is the same as the classical logarithmic score $\emph{LS}(f,y) =$ $ -\log\{f(y)\}$ up to a constant. The DSS is a proper scoring rule relative to the class of probability measures with finite second moment. It is strictly proper relative to any class of probability measures that are characterized by their first two moments, such as Gaussian measures or other location-scale families of distributions \citep{GneitingRaftery2007}.

We assume now that mean and variance of the predictive distribution $F$ match mean and variance of the outcome $Y$. The following properties of the CRPS and the DSS can be found in \citet{HeldRufibachETAL2010}. 
For the DSS we get
\begin{equation}\label{eq:EDSS}
  \E\{\emph{DSS}(F,Y)\} = \frac{1}{2} + \log(\sigma).
\end{equation}
If the distribution of $Y$ has finite fourth moment then $\Var\{\emph{DSS}(F,Y)\}$ is a constant that does not depend on $\mu$ or $\sigma$. If $Y$ has a normal distribution then $\Var\{\emph{DSS}(F,Y)\} = \frac{1}{2}$.
Similar results for the CRPS are harder to obtain. \citet{HeldRufibachETAL2010} show the following lemma.
\begin{lem}
Let $X_0$ be a random variable with finite second moment. For $a \in \mathbb{R}$ and $b > 0$, let $Y = a + bX_0$, let $F$ be the CDF of $Y$, and $\sigma^2$ its variance. Then,
\begin{equation}\label{eq:ECRPS}
  \E\{\emph{CRPS}(F,Y)\} = d \,\sigma \quad \text{and} \quad \Var\{\emph{CRPS}(F,y)\} = D \,\sigma^2,
\end{equation}
where
\[
d = \frac{\E|X_0-X_0'|}{2\sqrt{\Var(X_0)}} \quad \text{and} \quad D = \frac{\Var\big\{\E\big(\big|X_0 - X_0'|\,\big|X_0\big)\big\}}{\Var(X_0)}
\]
with $X_0'$ an independent copy of $X_0$.
\end{lem}
The lemma shows that for location-scale families of distributions the expected CRPS of an ideal forecast is proportional to the predictive standard deviation $\sigma$, and the variance of the CRPS is proportional to the predictive variance $\sigma^2$. For the family of normal distributions we have $d = 1/\sqrt{\pi}$ and $D = \{1/3 - (4-\sqrt{12})/\pi\} \approx 0.16275$. The constants for other families can be calculated at least numerically.

For the score regression approach, we consider $N$ independent and identically distributed observations $(F_{1,n}$, \ldots, $F_{k,n}$, $Y_n$, $V_n)$, $1 \le n \le N$ in the prediction space setting; see Definition \ref{def:prediction.space}. The expectation of the DSS depends on the logarithm of the predictive standard deviation; see equation \eqref{eq:EDSS}. Therefore, we stipulate a regression model of the form
\begin{displaymath}
  \emph{DSS}(F_{1,n},Y_n) = a + b_1 \log(\sigma_{1,n}) + \ldots + b_k \log(\sigma_{k,n}) + \epsilon_n,
\end{displaymath}
where $\sigma_{i,n}$ is the predictive standard deviation of $F_{i,n}$ and $\epsilon_n$ is an independent error with mean zero. Since the variance of the DSS is constant, irrespectively of the predictive variance, we can use a homoscedastic regression model to compute the least squares estimators $\hat{a}, \hat{b}_1, \ldots, \hat{b}_k$. In the case $k=1$ this is the same model as proposed at \citet[eq.~(7)]{HeldRufibachETAL2010}. We need to assume that the scores have finite variance, which is fulfilled if $Y$ has a finite fourth moment (conditional on $\mathcal{A}_1$). 

For the CRPS, motivated by \eqref{eq:ECRPS}, we stipulate the regression model
\begin{displaymath}
  \emph{CRPS}(F_{1,n},Y_n) = c + d_1 \sigma_{1,n} + \ldots + d_k \sigma_{k,n} + \epsilon_n.
\end{displaymath}
We have $\Var\{\emph{CRPS}(F_{1,n},Y_n)\} \varpropto \sigma_{1,n}$ and use a weighted
regression analysis with weights $1/\sigma_{1,n}$ to obtain estimators $\hat{c},\hat{d}_1,\dots,\hat{d}_k$; see for example \citet{MontgomerPeckETAL2001}. 

Both of these models can be used for testing if the forecaster $F_1$ is cross-ideal with respect to $\mathcal{A}_1 = \sigma(F_1), \ldots, \mathcal{A}_k=\sigma(F_k)$ in case of a normal forecaster $F_1$. The DSS can also be used if the prediction is non-normal as emphasized in \citet{HeldRufibachETAL2010}. The CRPS model is useful for any location-scale family of distributions.

We have the null hypotheses
\begin{align*}
  &H_0: a = 1/2, b_1 = 1 \mbox{ and } b_2 = \ldots = b_k = 0 \mbox{ for DSS},\\
  &H_0: c = 0, d_1 = 1/\sqrt{\pi} \mbox{ and } d_2 = \ldots = d_k = 0 \mbox{ for CRPS},
\end{align*}
and perform a $\chi^2$-test. We use the test statistics
\begin{align*}
  T_{DSS} &= (\hat{a} - 1/2, \hat{b}_1 - 1, \hat{b}_2, \ldots, \hat{b}_k) \hat{\Sigma}_{DSS}^{-1}
    (\hat{a} - 1/2, \hat{b}_1 - 1, \hat{b}_2, \ldots, \hat{b}_k)^\top,\\
  T_{CRPS} &= (\hat{c}, \hat{d}_1 - 1/\sqrt{\pi}, \hat{d}_2, \ldots, \hat{d}_k) \hat{\Sigma}_{CRPS}^{-1}
    (\hat{c}, \hat{d}_1 - 1/\sqrt{\pi}, \hat{d}_2, \ldots, \hat{d}_k)^\top,
\end{align*}
where $\hat{\Sigma}_{DSS}$, $\hat{\Sigma}_{CRPS}$ are the estimated covariance matrices. Both test statistics, $T_{DSS}$ and $T_{CRPS}$, are
asymptotically $\chi^2$-distributed with $1 + k$ degree of freedom and asymptotic
$p$-values are given by
\begin{displaymath}
  \pi_{DSS} = 1 - \chi^2_{1 + k}(T_{DSS}), \quad\text{and}\quad \pi_{CRPS} = 1 - \chi^2_{1 + k}(T_{CRPS}).
\end{displaymath}
If $k=1$, that is the case of just one forecaster, we obtain the test for an ideal forecaster suggested in
\citet{HeldRufibachETAL2010}. We call the tests presented in this section \emph{SRA tests} as they are based on a score regression approach (SRA). As noted already by \citet{HeldRufibachETAL2010}, SRA tests can only be used if each forecaster predicts at least two different variances, therefore we cannot apply it to Example \ref{ex:marginal.calibration} by \citet{GneitingRanjan2013}. Instead, we consider the following setup for illustration.
 
\begin{ex}
  \label{ex:cross.ideal.test.crps}
We consider two forecasters $F_{1,n} = \cN(0,(1 + \sigma_n)^2)$ and $F_{2,n} = \cN(0,(1 +
\sigma_n + \epsilon_n)^2)$, where $\sigma_n \sim \mathcal{U}([0,1])$ and $\epsilon_n \sim
\cN(0,1/16)$. The observations are $Y_n \sim \cN(0,(1 + \sigma_n)^2)$. It is clear, that $F_1$
is ideal with respect to $\mathcal{A}_1 = \sigma(\sigma_n)$, $F_2$ is not ideal with respect to $\mathcal{A}_2 = \sigma(\sigma_n,\epsilon_n)$, $F_1$ is cross-ideal
with respect to $F_1, F_2$, and $F_2$ is not cross-ideal with respect to $F_1, F_2$.
In Table \ref{tab:cross.ideal1}, the Monte-Carlo powers of the CRPS tests are displayed. The tests are performed at significance level $\alpha= 0.05$. The results are in accordance with the theoretical considerations. However, the Monte-Carlo power of the test if $F_2$ is cross-ideal with respect to $F_1, F_2$ is higher then in the test if $F_2$ is ideal with respect to $F_2$. It is interesting to see that taking $F_1$ into account helps to detect that $F_2$ is not ideal with respect to $F_2$.
\end{ex}

 \begin{table}
    \centering
    \begin{tabular}{|l|ccccc|}\hline
      $N$ & 30 & 50 & 100 & 200 & 500\\\hline
      $F_1$ wrt$F_1$ & 0.083 & 0.073 & 0.061 & 0.056 & 0.050\\
      $F_1$ wrt$F_1,F_2$ & 0.094 & 0.073 & 0.064 & 0.059 & 0.055\\
      $F_2$ wrt$F_2$ & 0.240 & 0.285 & 0.429 & 0.707 & 0.962\\
      $F_2$ wrt$F_1,F_2$ & 0.292 & 0.376 & 0.582 & 0.852 & 0.998\\\hline
   \end{tabular}
    \caption{Monte-Carlo powers for the CRPS tests with
      sample sizes $N$ and 10'000 simulations described in detail in Example \ref{ex:cross.ideal.test.crps}.\label{tab:cross.ideal1}}
  \end{table}

\begin{ex}[Example \ref{ex:cross-ideal} continued]
  \label{ex:cross.ideal.test.dss}
Considering again Example \ref{ex:cross-ideal}, we used the DSS test to assess if the forecasters are cross-ideal. In Table \ref{tab:cross.ideal2} we present the Monte-Carlo powers of the tests which are performed at significance level $\alpha=0.05$. The CRPS cannot be used for $F_2$ since the forecast is not normal. As expected, the tests show that $F_1$ is ideal with respect to $\mathcal{A}_1$ and also cross-ideal with respect to $\mathcal{A}_1,\mathcal{A}_2$. For a sample size of $N=200$ the level of the test is kept reasonably well. The forecaster $F_2$ is ideal with respect to $\mathcal{A}_2$ but fails to be cross-ideal with respect to $\mathcal{A}_1,\mathcal{A}_2$. The test shows a good power already for a sample size of $N=100$. However, for $F_2$ the test is slightly anticonservative even for a sample size of $N=500$.
\end{ex}

  \begin{table}
    \centering
    \begin{tabular}{|l|ccccc|}\hline
      $N$ & 30 & 50 & 100 & 200 & 500\\\hline
      $F_1$ wrt$F_1$ & 0.085 & 0.075 & 0.064 & 0.055 & 0.056\\
      $F_1$ wrt$F_1,F_2$ & 0.084 & 0.072 & 0.065 & 0.057 & 0.054\\
      $F_2$ wrt$F_2$ & 0.111 & 0.096 & 0.090 & 0.074 & 0.072\\
      $F_2$ wrt$F_1,F_2$ & 0.286 & 0.416 & 0.718 & 0.963 & 1.000\\\hline
   \end{tabular}
    \caption{Monte-Carlo power for the DSS tests with
      sample sizes $N$ and 10'000 simulations described in detail in Example \ref{ex:cross.ideal.test.dss}.\label{tab:cross.ideal2}}
  \end{table}

\subsection{Summary}\label{sec:summary}

We have presented three different approaches for testing cross-calibration, the CEP tests in Section \ref{sec:condex}, the LRA tests in Section \ref{sec:prob.cross.cal.test} and the SRA tests in Section \ref{sec:cross.ideal.test}. 
The first two approaches allow to test for cross-calibration of $F_{1}$ with respect to any subset $J \subset \{1,\dots,k\}$, whereas the SRA tests only allow to test for $F_{1}$ being cross-ideal which is equivalent to requiring that $1 \in J$. The CEP test and the LRA test with respect to $J=\emptyset$ are tests for probabilistic calibration, that is, the classical hypothesis of uniformity and independence of PIT values. While the SRA tests require independent forecast-observation tuples,  the CEP and the LRA tests are formulated in a prediction space for serial dependence, which is a scenario that is frequently encountered in practice; see also Section \ref{sec:data}. 

The CEP test has the advantage that it provides information concerning the parts of the distribution where miscalibration is detected (in terms of quantile levels); this is illustrated in Figures \ref{fig:dat1} and \ref{fig:dat2}. It may be considered a disadvantage that the adjusted $p$-values are simulation based and depend on a grid $0 < z_1 < \dots < z_M < 1$ that is to be chosen. In simulations, the method has shown to be robust to the number $M$ of grid points. 
On the contrary, the $p$-values for the LRA test are given explicitly. The forecasters have to be described through a finite-dimensional parameter vector and there are some restrictions concerning the predictive parameters, as it has to be ensured that the design matrix $\mathbf{D_J}$ at \eqref{eq:design1} has full rank. For the forecasters of Example \ref{ex:marginal.calibration}, the LRA test has overall a better power than the CEP test; see Examples \ref{ex:CEP-GN} and \ref{ex:LRA.Gn}. The difference is minor, except for the hypothesis that the forecaster $F_3$ is ideal. Here, for sample size $N=50$, the LRA test achieves a power of $0.734$, whereas the CEP test only has a power of $0.168$. For the forecasters in Example \ref{ex:cross-ideal} the CEP test outperformed the LRA test; see Examples \ref{ex:CEP-CI} and \ref{ex:prob-cross-ideal-test}. In fact, for sample size $N=200$, the power of the CEP test is more than three times higher than the power of the LRA test. 

The following modifications of the CEP and the LRA tests are straight forward but unexplored. The logistic regression model in \eqref{eq:logitmodel} can be replaced by any other regression model for a binary outcome variable, where it is possible to formulate a test for an analogous  pointwise null hypothesis as given at \eqref{eq:pointwisenull}. If forecasters choose their distributions from a parametric class of distributions as assumed in LRA approach, it could also be considered to regress the random variables $B_z = \one(Z_{1,t} \le z)$ on the predicted parameter values. In the LRA, the linear regression model stipulated at \eqref{eq:linmod} can be replaced by some other regression model for a vector of real valued outcomes. 

We would like to remark that the CEP and the LRA tests are formulated in the prediction space setting for serial dependence and make use of condition \eqref{eq:serialass}. It appears that deciding whether this assumption is justified in a given application context is sometimes a delicate matter. For example, if a forecaster $i$ bases her predictions purely on intuition, then \eqref{eq:serialass} is certainly justified. If a forecaster $j$ uses a time series model for predictions, that is, predictions are exclusively derived from past data, then one may argue that assumption \eqref{eq:serialass} fails and the CEP and LRA tests should only be applied with respect to sets $J$ such that $j \not\in J$. It may be that some parameters of a predictive distribution are derived from past data, whereas others are from external sources such as expert opinion. Here, it could be argued that one should only regress on the latter type of parameters in the LRA tests and use a regression model in terms of these parameters for the CEP tests. A different point of view would be that the parameters based on past data are derived through a subjectively chosen model, thus after the fitting procedure they should rather be viewed as personal opinion of the forecaster than as an information influencing the outcome. We will discuss condition \eqref{eq:serialass} further in Section \ref{sec:data}.

The SRA tests, based on the score regression approach, require independent forecast-observation tuples, or, more precisely, independent sequences of realized score values $\CRPS(F_{1,n},Y_n)$ or $\DSS(F_{1,n},Y_n)$, $1 \le n \le N$, which may be a weaker requirement. They are asymptotic tests, that appear to be working well for sample sizes of at least $N=100$; see Tables \ref{tab:cross.ideal1} and \ref{tab:cross.ideal2}. The SRA test with the CRPS works only for predictive distributions from one location-scale family, whereas the SRA test with the DSS requires only that the predictive distributions have finite fourth moments. In both cases, the predictive standard deviations have to differ for at least two observations. For the forecasters of Example \ref{ex:cross-ideal}, the SRA test with the DSS showed a better power than the CEP test, so it is an interesting alternative despite the more restrictive assumptions; see Examples \ref{ex:CEP-CI} and \ref{ex:cross.ideal.test.dss}. In particular, for sample size $N=200$ the SRA test had a power of $0.963$ detecting that $F_2$ is not cross-ideal with respect to $F_1,F_2$, whereas the CEP test had a power of $0.464$.

In the case of independent forecast-observation tuples it is possible to derive a test for marginal cross-calibration by testing for mean zero in \eqref{eq:mar-cal-plot} for each $y \in \R$. It has turned out in simulations, that the resulting asymptotic test has several problems for applications. For completeness, we report these findings in Appendix \ref{sec:marg.cross.cal.test}.

\section{Data example}\label{sec:data}
The Bank of England (BoE) predicts the inflation rate of every quarter by using a probabilistic forecast with a potentially asymmetric two-piece normal distribution with parameters $\mu \in \R$ and $\sigma_1, \sigma_2 > 0$ and density
\begin{equation} \label{eq:twopiece}
  f(y) = 
  \begin{cases}
    \left(\frac{\pi}{2}\right)^{-1/2} (\sigma_1 + \sigma_2)^{-1} \exp\left( - \frac{(y- \mu)^2}{2\sigma_1^2}\right) & \text{if } y \leq \mu,\\
     \left(\frac{\pi}{2}\right)^{-1/2} (\sigma_1 + \sigma_2)^{-1} \exp\left( - \frac{(y- \mu)^2}{2\sigma_2^2}\right) & \text{if } y \geq \mu.\\   
  \end{cases}
\end{equation}
The forecasts have been issued by the BoE's Monetary Policy Committee since February 1996 for the first quarter of 1996 and are publicly available online. The first quarter is from March to May, the second quarter from June to August, and so forth. Furthermore, there are forecasts available which have been issued between February 1993 and May 1997. These were converted into density forecasts retrospectively. Until the first quarter of 2004, the forecasts have been issued to predict RPIX inflation rates. But since the first quarter of 2004, inflation has been predicted and assessed in terms of percentage changes over twelve months of the CPI. The observed RPIX as well as the CPI inflation rates are available from the Office for National Statistics under codes CDKQ and D7G7, respectively. There is no simple transformation that converts an RPIX inflation rate into a CPI inflation rate and vice versa, so we have analysed the two data sets separately; RPIX inflation rate predictions from the first quarter of 1993 to the first quarter of 2004 and CIP inflation rate predictions from the first quarter of 2004 to the first quarter of 2015. In both cases we have 45 forecast-observation tuples. For further detail on the data set, see \citet[Section 4.1]{GneitingRanjan2011}. The BoE inflation forecasts have also been previously analysed for example by \citet{Wallis2003,Clements2004,MitchellHall2005,GalbraithNorden2012}.

For both data sets, we compare the BoE predictions with a Gaussian autoregression (AR) of order one with rolling estimation window of length six quarters, which leads to Gaussian density forecasts. The prediction horizon we consider is one quarter. As discussed in Section \ref{sec:summary}, the CEP and LRA tests make use of condition \eqref{eq:serialass}. While we believe that the BoE forecasts can be assumed to satisfy \eqref{eq:serialass}, it is more debatable in the case of the AR forecasts. If one is not willing to believe that \eqref{eq:serialass} holds in this case, one should only consider the CEP and the LRA tests with respect to the empty set, that is probabilistic calibration, and cross-calibration with respect to BoE. The conclusions we can draw about the quality of the forecasts remain essentially the same. Due to the serial dependence in the data, we do not apply the SRA tests.

First, we consider the CEP tests. The results for the BoE density forecasts can be seen in Figure \ref{fig:dat1} and the ones for the AR forecasts in Figure \ref{fig:dat2}. In both plots the grid is  $z_m = \{1 + (148/149)m\}/150$ for $0 \le m \le 149$ and $20'000$ bootstrap replications are used to calculate the adjusted $p$-values under the null hypothesis. For the RPIX inflation rate forecasts in the top panel of Figure \ref{fig:dat1}, the BoE forecast seems to be probabilistically calibrated and also cross-calibrated with respect to the AR forecast. It fails to be ideal, that is cross-calibrated with respect to itself. As a theoretical consequence it also fails to be cross-calibrated with respect to both, the AR forecast and itself by Theorem \ref{thm:cross.individual}. The CEP test picks this up correctly, and rejects the null hypothesis with respect to BoE or with respect to BoE and AR for some small exceedance probabilities between zero and $0.05$. However, it should be remarked that the rejection region is small allowing the tentative conclusion that the BoE forecast is not far from being ideal or cross-ideal. For the CPI inflation rate predictions in the bottom panel of Figure \ref{fig:dat1}, the situation is different. Probabilistic calibration of the BoE forecast is rejected for exceedance probabilities between $0.13$ and $0.26$. Note that this result makes no use of assumption \eqref{eq:serialass}. Cross-calibration with respect to AR and with respect to BoE itself is also rejected in some parts of the region between $0.13$ and $0.26$. In this case, the CEP test is not able to pick up a failure of cross-calibration with respect to both, AR and BoE, although this is a theoretical consequence of the lack of probabilistic calibration by Theorem \ref{thm:cross.individual}.

\begin{figure}
  \centering
\includegraphics[scale=0.85]{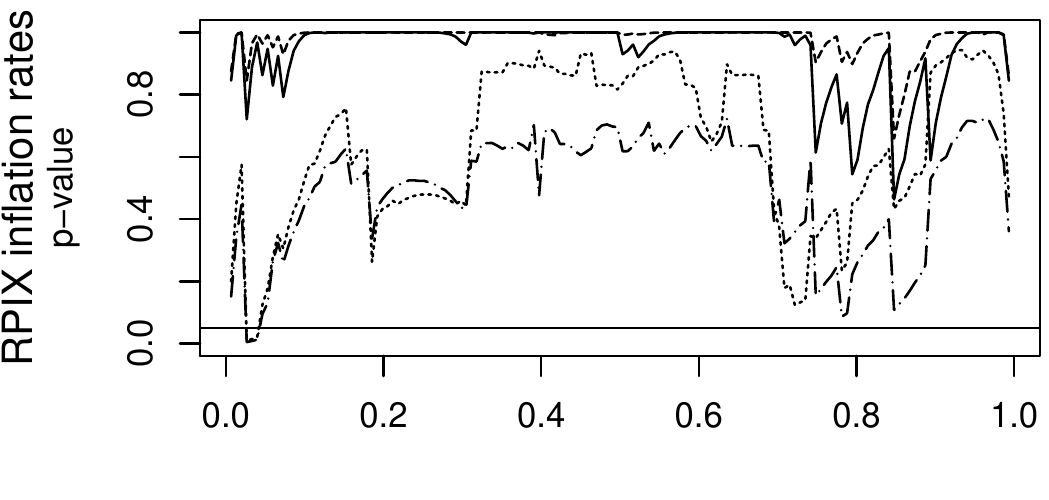}
\includegraphics[scale=0.85]{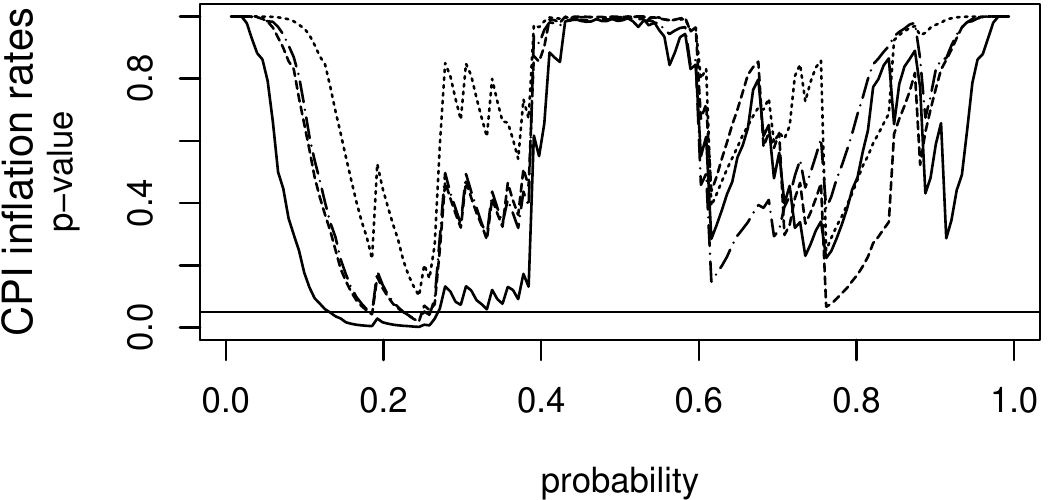}
\caption{The $p$-values of the CEP tests for the BoE forecast. The top panel corresponds to the prediction of RPIX inflation rates, whereas the bottom panel displays the results for CPI inflation rates. The solid horizontal lines give 0.05 level; the solid lines refer to probabilistic calibration (cross-calibration with respect to the empty set); the dashed lines refer to cross-calibration with respect to AR, the dash-dotted lines with respect to BoE and the dotted lines with respect to BoE and AR.\label{fig:dat1}}
\end{figure}

\begin{table}
  \centering
  \begin{tabular}{|l|cccc|}\hline
    RPIX& BoE wrt $\emptyset$ & BoE wrt AR & BoE wrt BoE & BoE wrt BoE, AR \\ \hline
    F-test & 0.338 & 0.185 & 0.010 & 0.010 \\
    AD-test & 0.496 & 0.589 & 0.822 & 0.731 \\
    adjusted & 0.676 & 0.370 & 0.021 & 0.020 \\ \hline
    CIP & BoE wrt $\emptyset$ & BoE wrt AR & BoE wrt BoE & BoE wrt BoE, AR \\ \hline
    F-test & 0.3973 & 0.5629 & 0.1486 & 0.2228\\
    AD-test & 0.0102 & 0.0122 & 0.0073 & 0.0047\\
    adjusted & 0.0203 & 0.0245 & 0.0146 & 0.0093\\ \hline
  \end{tabular}
  \caption{The $p$-values for the LRA tests for the BoE forecast.\label{tab:BoE.LRA1}}
\end{table}

\begin{figure}
  \centering
\includegraphics[scale=0.85]{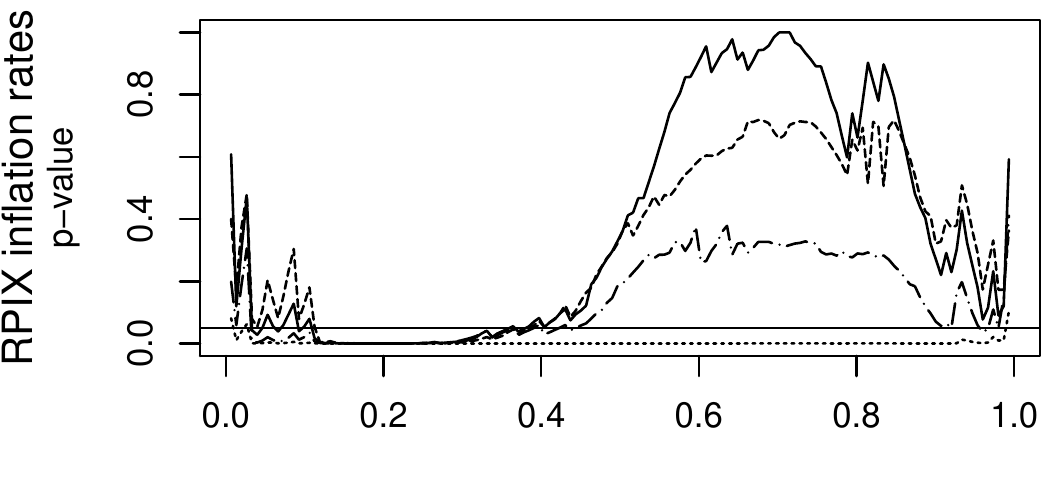}
\includegraphics[scale=0.85]{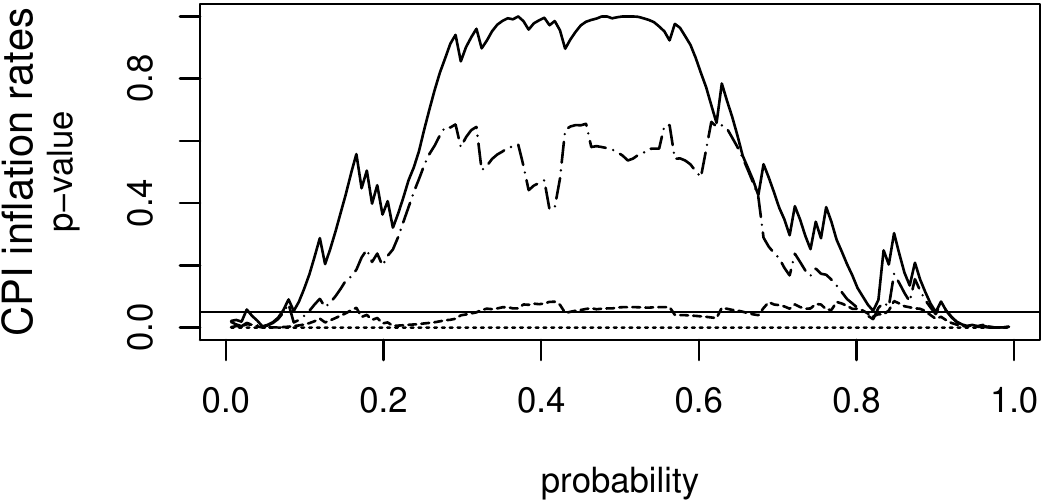}
\caption{The $p$-values of the CEP tests for the AR forecast. The top panel corresponds to the prediction of RPIX inflation rates, whereas the bottom panel displays the results for CPI inflation rates. The solid horizontal lines give the 0.05 level; the solid lines refer to probabilistic calibration (cross-calibration with respect to the empty set); the dashed lines refer to cross-calibration with respect to BoE, the dash-dotted lines with respect to AR and the dotted lines with respect to BoE and AR.\label{fig:dat2}}
\end{figure}

\begin{table}
  \centering
  \begin{tabular}{|l|cccc|}\hline
    RPIX & AR wrt $\emptyset$ & AR wrt AR & AR wrt BoE & AR wrt BoE, AR \\ \hline
    F-test & 0.193 & 0.136& 0.325 & $<$0.001\\
    AD-test & 0.003 & 0.027 & 0.025 & \phantom{$<$}0.570 \\
    adjusted & 0.006 & 0.054 &0.049 & $<$0.001 \\ \hline
    CIP & AR wrt $\emptyset$ & AR wrt AR & AR wrt BoE & AR wrt BoE, AR \\ \hline
    F-test &  0.5109 & 0.1668 & $<$0.0001 & $<$0.0001\\
    AD-test & 0.0001 & 0.0004 & \phantom{$<$}0.1784 & \phantom{$<$}0.5436\\
    adjusted & 0.0002 & 0.0008 & $<$0.0001 & $<$0.0001\\ \hline
    \end{tabular}
  \caption{The $p$-values for the LRA tests for the AR forecast.\label{tab:BoE.LRA2}}
\end{table}

According to the CEP test, the AR forecast for the RPIX inflation rate is not probabilistically calibrated and therefore also not cross-calibrated, ideal or cross-ideal; see the top panel of Figure \ref{fig:dat2}. For all tests, the forecaster fails in the region of exceedance probabilities lower than $0.4$ and near to $1$. Cross-calibration with respect to BoE and AR is rejected for all exceedance probabilities with a very low $p$-value. While the overall conclusions remain the same for the CPI inflation rate forecasts, the situation is somewhat different; see the bottom panel of Figure \ref{fig:dat2}. Cross-calibration with respect to BoE and with respect to AR and BoE is rejected for almost all exceedance probabilities. However, probabilistic calibration of the AR forecaster and cross-calibration with respect to itself is only rejected for some probabilities below $0.10$ and above $0.80$ indicating the the AR forecast might be superior to the BoE forecast for exceedance probabilities between $0.13$ and $0.26$.

Secondly, we consider is the LRA tests. The parametric class $\cF$ used for the tests is the class of two-piece normal distributions with parameters $\mu \in \R, \sigma_1 > 0, \sigma_2 > 0$ given at \eqref{eq:twopiece}. We can perform all the tests as for the CEP. The corresponding $p$-values can be found in Tables \ref{tab:BoE.LRA1} and \ref{tab:BoE.LRA2}. We also see if the estimated regression parameter failed to be zero or the standard normality assumption for the residuals was violated. The results coincide with the ones from the CEP, but we do not see in which region of exceedance probabilities the forecasters failed. On the other hand, in this application the LRA tests are consistent with Theorem \ref{thm:cross.individual} in the sense that rejection of cross-calibration with respect to a smaller set implies rejection with respect to any superset.


\section{Discussion}\label{sec:discussion}
We have extended the prediction space setting of \citet{GneitingRanjan2013} to accommodate serially dependent forecasts which are commonly encountered in practice. For prediction spaces with serial dependence, we have shown a refined version of the result of \citet{DieboldGuntherETAL1998} on uniformity and independence of PIT values. It relies on condition \eqref{eq:serialass}, whose implications should be studied in greater detail. We have focussed on the case of one period ahead forecasts like in the original result. As mentioned in Remark \ref{rem:qstep}, an analogous result continues to hold for $q$-step ahead forecasts. However, additional complications arise in testing for cross-calibration, which need further investigation in future research.

We have refined the notions of calibration to notions of \emph{cross}-calibration and we have provided powerful statistical tests for these properties requiring minimal assumptions on the sequences of forecasts and observations. The characterization of cross-calibration and cross-ideal forecasters in Proposition \ref{prop:pccci} sheds some light on the difference between ideal forecasters and probabilistically calibrated forecasters as discussed in \citet{GneitingRanjan2013}. It is remarkable to note that with our approaches, testing for ideal forecasters is not more difficult than testing for probabilistic calibration, contrary to the doubts voiced in \citet{GneitingRanjan2013}.

In order to optimize forecasting performance, it is natural to combine forecasts. \citet{GneitingRanjan2013} have proposed combination formulas and aggregation methods to combine several forecasters; see also \citet{RanjanGneiting2010}. It would be interesting to consider under which conditions, calibrated forecasters can be combined to yield cross-calibrated forecasts. Also, the more refined notions of cross-calibration in this paper, may help to identify which forecasters to include in combination formulas and which ones do not add additional information about the future outcome. Finally, combining forecasts is only a good idea if the predictions are based on different information sets. If there is a cross-calibrated forecaster with respect to all forecasters, any combination of forecasts would compromise on forecast quality. 

Our approach may add another perspective on the concerns raised by \citet{MitchellWallis2011} concerning the principle to ``Maximize sharpness subject to calibration'' formulated by \citet{MurphyWinkler1987,GneitingRaftery2007}. In fact, the concept of cross-calibration allows to assess the statistical compatibility of several forecasters with the observations. When considering calibration and sharpness as suggested by \citet{GneitingRaftery2007}, calibration concerns the interplay of one forecaster and the observation, whereas sharpness compares forecasters but makes no reference to observations. Possibly, the formulated guiding principle should be modified to ``Maximize sharpness subject to cross-calibration'' which is a stronger requirement in terms of calibration and therefore gives somewhat less importance to sharpness.

\section*{Acknowledgements}
The authors would like to thank Tilmann Gneiting for suggesting the data example. 

\bibliographystyle{plainnat}
\bibliography{Literature}

\appendix

\section{Calculations for Example \ref{ex:marginal.calibration}}
\label{appendix:marg.cross.cal}

Let $\mu \sim \cN(0,1)$ and let $\tau$ takes values $1$ or $-1$ with equal probability independent of $\mu$. Conditional on $\mu$ and $\tau$, the observation is $Y \sim \cN(\mu,1)$ and the forecasters have the following predictive distribution functions:
\begin{align*}
  F_1(y) &= \Phi(y - \mu),\\
  F_2(y) &= \Phi\left(\frac{y}{\sqrt{2}}\right),\\
  F_3(y) &= \frac{1}{2} \Phi(y - \mu) + \frac{1}{2} \Phi(y- \mu
  -\tau),\\
  F_4(y) &= \Phi(y+\mu)
\end{align*}
for $y \in \mathbb{R}$.  As in \citet{GneitingBalabdaouETAL2007}, we use the definitions $\Psi_+(x) = \frac{1}{2} \{\Phi(x) + \Phi(x-1)\}$, $\Psi_-(x) = \frac{1}{2} \{\Phi(x) + \Phi(x+1)\}$. Thus, $\Psi_-(x) = \Psi_+(x+1)$ and $\Psi_-^{-1}\{\Psi_+(x+1)\} = x$. 
\begin{prop}
The unfocused forecaster $F_3$ is probabilistically cross-calibrated with respect to $F_1,F_2,F_4$.
\end{prop}
\begin{proof}
Let $y \in (0,1)$. We have
\begin{align*}
\mathbb{Q}(Z_{F_3} \le y | F_1, F_2, F_4) &= \mathbb{Q}(Z_{F_3} \le y | \mu) \\
&= \frac{1}{2} \mathbb{Q}\{\Psi_+(Y-\mu) \le y | \mu\} + \frac{1}{2} \mathbb{Q}\{\Psi_-(Y-\mu) \le y | \mu\}\\
&= \frac{1}{2} \Phi\{\Psi_+^{-1}(y)\} + \frac{1}{2}\Phi\{\Psi_-^{-1}(y)\} = y.
\end{align*}
\end{proof}

\section{Testing for marginal cross-calibration}
\label{sec:marg.cross.cal.test}

We consider two forecasters $F_1$ and $F_2$ within the prediction space
setting. Our interest lies in
\begin{displaymath}
  S(y) = F_2(y) - \one\{F_2^{-1}(Z_{F_1}^Y) \le y\}, \quad y \in \supp(Y),
\end{displaymath}
where $\supp(Y)$ denotes the support of of the observation $Y$. We would like to test if $\E_{\Q}{S(y)} = 0$ for all $y \in \supp(Y)$, because
this is equivalent to marginal cross-calibration of  $F_1$ with respect to $F_2$; cf.~Definition \ref{def:cross.calibration}. 

We suppose that we have $N$ independent and identically distributed observations $(F_{1,n},F_{2,n},Y_n,V_n)$ for $1 \le n \le N$ in a prediction space and define for each $n$,
\begin{displaymath}
  S_n(y) = F_{2,n}(y) - \one\{F_{2,n}^{-1}(Z_{1,n}) \le y\}, \quad y \in \supp(Y).
\end{displaymath}
We pick a sequence $y_0 < y_1 < \ldots < y_m$ in the support
of $Y$ and define
\begin{displaymath}
  \mathbf{S_n} = (S_n(y_0),S_n(y_1), \ldots, S_n(y_m))^T,
\end{displaymath}
 and $\mathbf{\bar{S}_N} = (1/N)\sum_{n=1}^N \mathbf{S_n}$. Let $\Sigma_N = (1/N)\sum_{n=1}^N (\mathbf{S_n} - \mathbf{\bar{S}_N})(\mathbf{S_n} - \mathbf{\bar{S}_N})^T$ be the sample covariance matrix. 
If $F_1$ is marginally cross-calibrated with respect to $F_2$, then $\mathbb{E}(\mathbf{S_n}) = 0$. Therefore, by standard arguments of probability theory, the test statistic $T = N\, \mathbf{\bar{S}_N}^T \Sigma_N^{-1} \mathbf{\bar{S}_N}$ converges in distribution to $\chi^2_m$, a chi-squared distribution with $m$ degrees of freedom. 

We test the null hypothesis that $\mathbb{E}_{\mathbb{Q}}S(y) = 0$ for all $y \in \supp(Y)$ for one particular finite distribution of $S(y)$. Therefore, the sequence $y_1, \ldots, y_m$ has to be chosen carefully. Simulations indicate that the level and power of the test is not much affected by the choice of $y_1,\dots,y_m$. However, for small sample sizes the number of grid points $m$ should be rather small, otherwise the sample covariance matrix may be singular and can not be inverted to compute the test statistic. Another reason for singularity of the sample covariance matrix for small sample sizes may be the choice of a grid point $y$  such that the probability that $F_i^{-1}\{Z_{F_{j}}^Y\} \le y$ is small. Unfortunately, for an individual test case, different choices of $y_1,\dots,y_m$ may lead to completely different $p$-values, which makes the test useless in practice. We illustrate these effects in the following example.

\begin{ex}
  \label{ex:marg.cross.cal.test}
We consider the forecasters $F_1,\ldots,F_4$ and the observation $Y$ from Example \ref{ex:marginal.calibration}. 
Let $N=500$ be the number of observations from $(F_i,F_j, Y)$, for each pair $F_i$ and $F_j$ with $1 \le i,j \le 4$. The results in Table \ref{tab:marg.cross.cal.test1} show that the marginal cross-calibration test performs well overall, and the performance is relatively unaffected by the choice of different grid points $y_1, y_2, \ldots, y_m$. However, if we consider an increasing number of grid points for the same data set the $p$-value changes substantially. This is illustrated in Figure \ref{fig:pvaluesGn} for five different simulated data sets with $N = 500$ and the null hypothesis that $F_3$ is marginally cross-calibrated with respect to $F_4$.
\end{ex}

\begin{table}
\begin{center}
\begin{tabular}{|c|cccc|}\hline
$m=9$ & $F_1$ & $F_2$ & $F_3$ & $F_4$\\\hline
$F_1$ & 0.066 & 0.0694 & 0.0655 & 0.0622\\
$F_2$ & 1 & 0.0671 & 1 & 1\\
$F_3$ & 0.0689 & 0.0708 & 0.5782 & 0.0597\\
$F_4$ & 1 & 1 & 1 & 0.0668\\\hline\hline
$m=4$ & $F_1$ & $F_2$ & $F_3$ & $F_4$\\\hline
$F_1$ & 0.0556 & 0.0545 & 0.0511 & 0.0578\\
$F_2$ & 1 & 0.0555 & 0.9972 & 1\\
$F_3$ & 0.0531 & 0.0534 & 0.5122 & 0.0563\\
$F_4$ & 1 & 1 & 1 & 0.0554\\\hline\hline
$m=3$ & $F_1$ & $F_2$ & $F_3$ & $F_4$\\\hline
$F_1$ & 0.0526 & 0.0568 & 0.0519 & 0.0545\\
$F_2$ & 0.9993 & 0.0524 & 0.986 & 1\\
$F_3$ & 0.0532 & 0.0566 & 0.431 & 0.0586\\
$F_4$ & 1 & 1 & 1 & 0.0521\\\hline
\end{tabular}
  \caption{Monte-carlo powers of the marginal cross-calibration tests for sample size $N=500$
    and grid points $(y_1, y_2, \ldots, y_9) = (-1.81, -1.19, -0.74, -0.36,  0.00,
    0.36,  0.74,  1.19,  1.81)$, $(y_1, y_2, y_3, y_4) = (-1.19, -0.35,0.35,1.19)$, $(y_1, y_2, y_3) = (-0.95,0,0.95)$, respectively, for the first, second and third table. The value in the $i$-th row and $j$-th column is the percentage of rejections of the null hypothesis that $F_i$ is marginally cross-calibrated with respect to $F_j$ at level $\alpha = 0.05$ for the forecasters $F_1, F_2, F_3$ and $F_4$ in Example \ref{ex:marg.cross.cal.test} in 10'000 simulations.\label{tab:marg.cross.cal.test1}}
\end{center}
\end{table}

\begin{figure}
  \centering
\includegraphics[width = 0.9\textwidth]{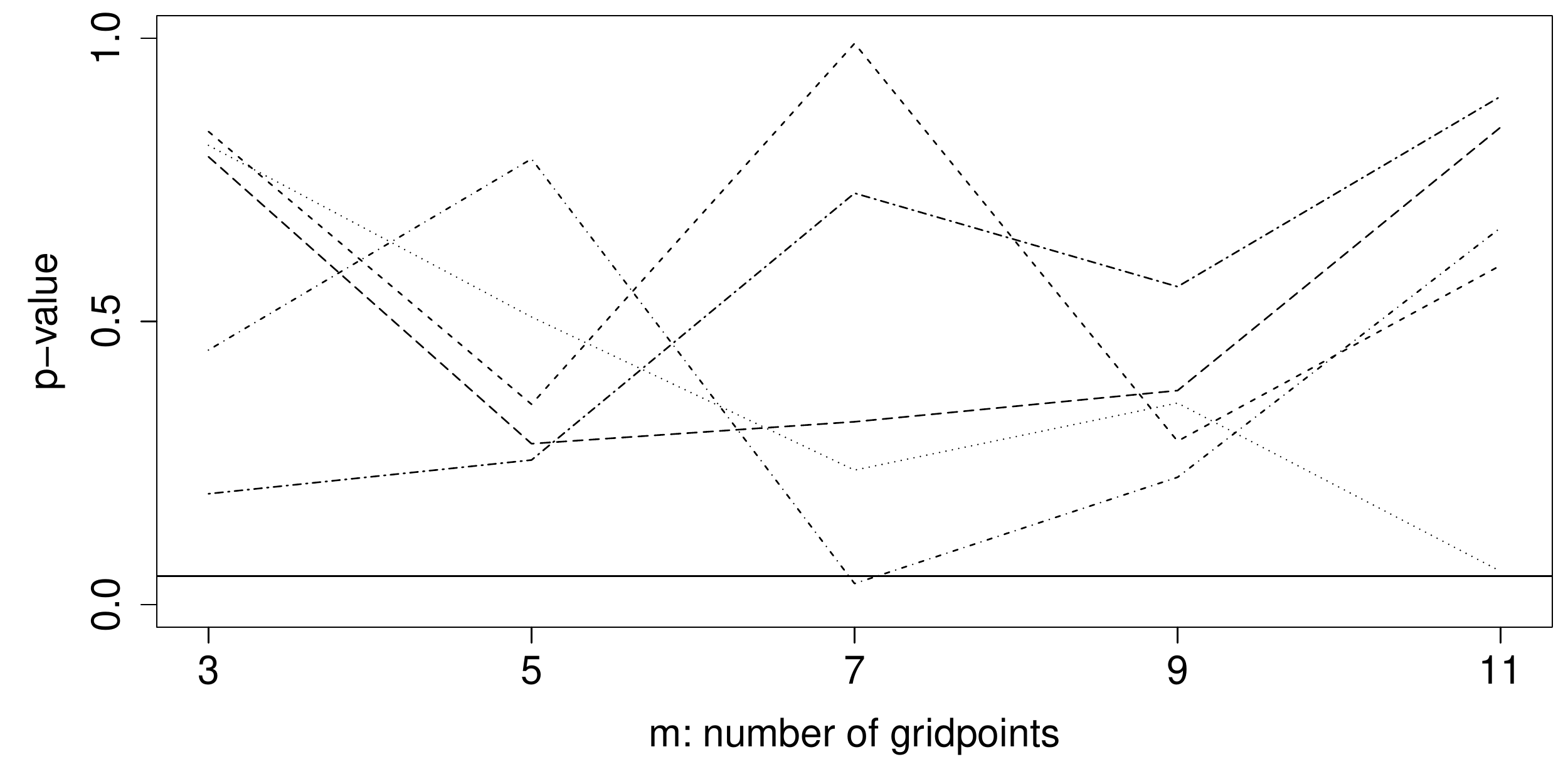}
\caption{The $p$-values of the marginal cross-calibration test for five different simulated data sets and an increasing number of grid points. The horizontal line marks the $\alpha=0.05$ significance level.\label{fig:pvaluesGn}}
\end{figure}

\end{document}